\documentclass[acmsmall,screen]{acmart}
\AtBeginDocument{%
  }

\setcopyright{none}
\copyrightyear{2018}
\acmYear{2018}
\acmDOI{XXXXXXX.XXXXXXX}
\usepackage{booktabs}   
\usepackage{graphicx}
\usepackage{subcaption} 
\usepackage{caption} 
\usepackage{comment}
\usepackage{hyperref}
\usepackage{bussproofs}
\usepackage{float}
\usepackage{tikz}
\usepackage{xspace}
\usepackage{xcolor}
\usepackage{pgfplots}
\pgfplotsset{compat=1.18}
\usepackage{geometry}

\usetikzlibrary{automata, positioning, arrows,calc, positioning, shapes, arrows.meta}
\usepackage{amsthm}
\usepackage[longend,ruled,linesnumbered]{algorithm2e}
\newtheorem{definition}{Definition}
\usepackage{amsmath}

\usepackage{amssymb}

\newif\ifshowcomments
\showcommentstrue   

\newcommand{\zak}[1]{\ifshowcomments{\color{red}Zak: #1}\fi}
\newcommand{\han}[1]{\ifshowcomments{\color{orange}Han: #1}\fi}
\newcommand{\dpw}[1]{\ifshowcomments{\color{blue}Dave: #1}\fi}

\newcommand{\denote}[1]{\ensuremath{\llbracket\texttt{$#1$}\rrbracket}}

\newcommand{\dpn}[1]{\ensuremath{\mathsf{PN}\llbracket\texttt{$#1$}\rrbracket}}
\newcommand{\dq}[1]{\ensuremath{\mathsf{Q}\llbracket\texttt{$#1$}\rrbracket}}

\newcommand{\dpr}[1]{\ensuremath{\mathsf{RN}\llbracket\texttt{$#1$}\rrbracket}}
\newcommand{\drn}[1]{\dpr{#1}}
\newcommand{\dw}[1]{\ensuremath{\mathsf{WN}\llbracket\texttt{$#1$}\rrbracket}}
\newcommand{\dwn}[1]{\ensuremath{\mathsf{WN}\llbracket\texttt{$#1$}\rrbracket}}
\newcommand{\dpkr}[1]{\ensuremath{\mathsf{PkR}\llbracket\texttt{$#1$}\rrbracket}}
\newcommand{\dpred}[1]{\ensuremath{\mathsf{Pred}\llbracket\texttt{$#1$}\rrbracket}}
\newcommand{\power}[1]{\ensuremath{\mathcal{P}(#1)}}

\newcommand{\OMIT}[1]{}
\newcommand{\defeq}{\triangleq}

\newcommand{\eqby}[3]{\mathrel{\overset{#1}{\underset{#2}{#3}}}}

\newcommand{\PR}{\ensuremath{PkR}} 

\newcommand{\pred}{\ensuremath{\textit{pred}}}

\newcommand{\pn}{\ensuremath{PN}}

\newcommand{\pk}{\ensuremath{pk}}

\newcommand{\trace}{\ensuremath{tr}}
\newcommand{\trcat}{\ensuremath{~}}

\newcommand{\pand}{\ensuremath{\cdot}}
\newcommand{\por}{\ensuremath{+}}
\newcommand{\pnot}{\ensuremath{\neg}}
\newcommand{\ptrue}{\ensuremath{1}}
\newcommand{\pfalse}{\ensuremath{0}}

\newcommand{\field}{\ensuremath{f}}
\newcommand{\const}{\ensuremath{c}}
\newcommand{\val}{\ensuremath{\mathit{val}}}
\newcommand{\srcip}{\ensuremath{\texttt{src.ip}}}
\newcommand{\dstip}{\ensuremath{\texttt{dst.ip}}}

\newcommand{\locf}{\ensuremath{\texttt{loc}}}

\newcommand{\uppkr}[2]{\ensuremath{#1 \leftarrow #2}}

\newcommand{\lift}[1]{\ensuremath{\overline{#1}}}
\newcommand{\havoc}{\ensuremath{\textit{havoc}}}

\newcommand{\idpkr}{\ensuremath{\lift{1}}}
\newcommand{\empkr}{\ensuremath{\lift{0}}}
\newcommand{\comppkr}{\ensuremath{\circ}}
\newcommand{\crosspkr}{\ensuremath{\times}}
\newcommand{\andpkr}{\ensuremath{\cap}}
\newcommand{\orpkr}{\ensuremath{\cup}}
\newcommand{\notpkr}{\ensuremath{\neg}}

\newcommand{\kand}{\ensuremath{\circ}}
\newcommand{\kor}{\ensuremath{+}}
\newcommand{\kcap}{\ensuremath{\cap}}
\newcommand{\kdiff}{\ensuremath{\backslash}}
\newcommand{\ks}[1]{\ensuremath{{#1}^*}}
\newcommand{\kdup}{\ensuremath{\mathit{dup}}}
\newcommand{\alltraces}{\ensuremath{\textit{alltraces}}}

\newcommand{\rapp}{\triangleright}

\newcommand{\rfilter}[1]{\ensuremath{\mathit{Filter}(#1)}}
\newcommand{\idname}{\ensuremath{\mathit{Id}}}
\newcommand{\rid}[1]{\ensuremath{\idname(#1)}}
\newcommand{\mapname}{\ensuremath{\mathit{Map}}}
\newcommand{\rmap}[2]{\ensuremath{\mapname(#1,#2)}}
\newcommand{\rdelete}[1]{\ensuremath{\mathit{Delete}(#1)}}
\newcommand{\rinsert}[1]{\ensuremath{\mathit{Insert}(#1)}}

\newcommand{\ror}{\ensuremath{+}}
\newcommand{\rconcat}{\ensuremath{\cdot}}
\newcommand{\rs}[1]{\ensuremath{{#1}^*}}

\newcommand{\wand}{\ensuremath{\circ}}
\newcommand{\wor}{\ensuremath{+}}
\newcommand{\wtimes}{\ensuremath{\otimes}}
\newcommand{\wres}[2]{\ensuremath{#1\restriction #2}}

\newcommand{\qs}[2]{\ensuremath{\mathit{Select}(#1,#2)}}

\newcommand{\ipred}[1]{\ensuremath{\mathbb{I}_{Pred}(\texttt{$#1$})}}
\newcommand{\ipkr}[1]{\ensuremath{\mathbb{I}_{PkR}(\texttt{$#1$})}}
\newcommand{\ipn}[1]{\ensuremath{\mathbb{I}_{PN}(\texttt{$#1$})}}
\newcommand{\irn}[1]{\ensuremath{\mathbb{I}_{RN}(\texttt{$#1$})}}
\newcommand{\iwn}[1]{\ensuremath{\mathbb{I}_{WN}(\texttt{$#1$})}}

\begin{document}

\title[Network Analysis with Parametric NetKAT]{Network Analysis with Parametric NetKAT}         


\author{Han Xu}
\email{hx3501@princeton.edu}          
\orcid{0000-0002-2548-6866}

\affiliation{
  \institution{Princeton University}            
  \country{United States}                    
}

\author{Zachary Kincaid}
\email{zkincaid@cs.princeton.edu}          
\orcid{0000-0002-7294-9165}

\affiliation{
  \institution{Princeton University}            
  \country{United States}                    
}

\OMIT{
\author{Ratul Mahajan}
\email{ratul@cs.washington.edu}          
\orcid{0009-0005-8005-6948}

\affiliation{
  \institution{University of Washington}            
  \country{United States}                    
}
}

\author{David Walker}
\email{dpw@princeton.edu}          
\orcid{0000-0003-3681-149X}

\affiliation{
  \institution{Princeton University}            
  \country{United States}                    
}

\begin{abstract}
Network engineers often need to perform network diagnosis and inference tasks, which frequently require answers to \emph{enumeration questions} such as ``\emph{Which packets} from the Internet arrive at host \(C\)?'' or ``\emph{Which single-link failures} disconnect my network?''
Parametric NetKAT is a new domain-specific language that combines elements of NetKAT, Relational NetKAT, and Weighted NetKAT into a single system and extends them with parameters, allowing users to pose such enumeration questions directly over network models.
This paper presents the design and semantics of Parametric NetKAT and illustrates its utility through a series of examples.
It shows how to compile Parametric NetKAT into NetKAT automata, develops new algorithms for efficiently collecting satisfying valuations, and proves the correctness of these procedures.
Finally, it evaluates the performance of Parametric NetKAT on a collection of benchmarks drawn from industrial sources.
\end{abstract}




\maketitle
\section{Introduction}

Networks are the arteries of the modern online world, keeping people and businesses connected to essential services.
Unfortunately, they are also complex: difficult to manage, hard to modify reliably, and challenging to diagnose when failures occur.
Moreover, when networks go down, the economic and societal consequences can be enormous.
For example, in 2022, an outage at Rogers Canada~\cite{rogers-outage} left 12 million people without connectivity, prevented businesses nationwide from accepting debit transactions, disrupted access to 911 emergency services, delayed radiation therapy for cancer patients, and affected a wide range of government  activities.

One way to reduce the risk of such events is to deploy network-verification technologies such as NetKAT~\cite{netkat}.
NetKAT is both a domain-specific language for modeling networks and a specification language for expressing network properties.
To use a NetKAT system, an engineer invokes a modeling program that parses raw network configurations and produces a NetKAT model of the network's packet-forwarding behavior.
Once this semantic model has been constructed, a variety of yes-or-no verification questions can be posed:

\begin{enumerate}
\item Are all packets from the internet destined to port 80 blocked before they arrive at host C?
\item Do packets sent from A with destination IP address 10.0.0.0 arrive at host B?
\item Do all packets originating outside the network travel through firewalls (FW1 or FW2) prior to arriving at hosts A, B, or C? 
\end{enumerate}

Such questions are phrased as equations or inequations between NetKAT expressions (including the expression generated automatically via the modeling program). The expressions are then compiled into automata, and automata-theoretic decision procedures provide the answers.  Recently, several productive extensions to the paradigm have been proposed, including extensions that add probabilities~\cite{prob-netkat}, which may be used to assess properties such as probabilistic reliability, or more generally weights~\cite{acevedo26weight}, which allow analysis of quantities such latency.  In addition, NetKAT relations~\cite{han26relational} allow compact specification and verification of network changes.
Overall, NetKAT’s combination of programmability and expressive specification language, rigorous denotational semantics, and strong empirical results from a number of implementations~\cite{foster15coal,mcnetkat,Moeller24Katch,han25artifact} make it an attractive framework for network verification despite the presence of many alternatives~\cite{hsa,veriflow,netplumber,deltanet}---more discussion of related work may be found in Section~\ref{sec:relatedwork}.

While NetKAT has focused on \emph{yes-no verification questions}, these are not the only questions network engineers want to answer.  Engineers often need to engage in network \emph{diagnosis} or \emph{inference} tasks, which require answering \textit{enumeration questions}---i.e., find the set of all solutions to some problem of interest.
Such questions include:
\begin{enumerate}
\item Which packets from the internet destined to port 80 arrive at host C?
\item Which filtering rules in the Firewall are responsible for blocking traffic from A to B?
\item Which single-link failures disconnect my network?
\item Which packets follow paths that exceed the maximum allowed length?
\end{enumerate}

In this paper, we introduce \emph{Parametric NetKAT}, a new domain-specific language for answering such enumeration questions.  Parametric NetKAT builds on NetKAT, Relational NetKAT, and Weighted NetKAT, and extends these sublanguages with \emph{parameters}.  The Parametric NetKAT solver returns
all valuations of the parameters that satisfy a user's constraint system.
While expanding the set of questions that may be answered, Parametric NetKAT retains the many benefits familiar to other NetKAT languages:  Rich network modeling capabilities, compositional and modular design, compact specifications, a rigorous semantics, and efficient algorithms. Moreover, by integrating
three sublanguages (NetKAT, Relational NetKAT, Weighted NetKAT), we bring under one umbrella a wider range of
modeling and querying infrastructure than ever before.
For example, while "classic" NetKAT serves as the base network modeling language, Relational NetKAT may be used to select or transform elements of the modeled network, and also to insert parameters at points of interest. Such facilities may be combined with elements of Weighted NetKAT to measure and constrain quantitative properties such as cost, latency, reliability, or path length.  


\OMIT{
As with ordinary NetKAT, Parametric NetKAT relies on a modeling program to convert network configurations of interest into a NetKAT expression that represents their semantics. With that expression as a building block, an engineer may write queries to extract information
of interest.  Typically, such queries involve using
elements of   By combining 
NetKAT, Relational NetKAT and Weighted NetKAT in
a single overarching language extended with parameters,
we can use any or all of the features as needed.
\dpw{tried to motivate the combination...}}

In principle, enumeration problems can be solved by existing (Weighted, Relational) NetKAT decision procedures by simply issuing a verification query for each parameter valuation.  However, when the set of parameter valuations is large, this approach is infeasible.  We develop a suite of algorithms for solving parameter enumeration problems that have the same complexity as a single verification query (albeit for NetKAT expressions operating on a larger packet space).   By representing the packet space (and sets of parameter valuations) symbolically using binary decision diagrams, we obtain practical algorithms for Parametric NetKAT that scale to industrially relevant network sizes and parameter spaces.

To summarize, the key contributions of the work include:

\begin{itemize}
\item Design of parametric extensions to the NetKAT family of languages and their semantics. Integration of features from NetKAT, Relational NetKAT and Weighted NetKAT into one system.

\item Development of a programming/specification methodology that uses of Relational NetKAT for parameter insertion and leads to compact specifications independent of network size.  An illustration of the utility of parametric extensions through a range of practical examples.
\item Design, implementation, and proof of correctness of algorithms for answering parametric NetKAT queries, including a compilation scheme from Parametric NetKAT to (non-parametric) NetKAT automata and a suite of parameter enumeration algorithms, capable of computing the set of parameter valuations under which (a boolean combination of) emptiness, equivalence, and quantitative verification queries hold.

\item Evaluation of the system on a collection of benchmarks drawn from other sources ranging from (a) an analysis of changes made in an industrial networks~\cite{Xieyang24Rela}, (b) the Topology Zoo~\cite{TopologyZoo} and (c) case studies from the Amazon’s Batfish tool~\cite{batfish}.
\end{itemize}  

\section{Background:  NetKAT and Friends}\label{sec:bag}

\paragraph*{NetKAT}
NetKAT~\cite{netkat} is a domain-specific language for describing the behavior of network data planes.
Our presentation of NetKAT follows the more recent and more expressive formulation developed in work on Relational NetKAT~\cite{han26relational}.
Each NetKAT expression denotes a set of \emph{packet traces} (\trace),
where a packet trace is a sequence of two or more \emph{located packets} (\pk).  A located packet is a record that includes
fields for standard packet headers such as source and destination IP address (\srcip, \dstip) as well as
the packet's current location in the network (\locf).
For example, a trace $\pk_1 \ \pk_2 \ \cdots \ \pk_n$ where
$\pk_1.\locf = A$, $\pk_1.\dstip = 10.0.0.0$, and
$\pk_n.\locf = C$, describes the fact that a packet 
with destination IP address $10.0.0.0$, starting at location
$A$ travels some number of hops to
location $C$.  Packet headers can be modified
by the network along the way, so for instance $\pk_n.\dstip$ might
be $10.0.0.1$.  

Two packet traces may be \emph{concatenated} when the last packet of one trace equals the first packet of the next trace; the equal packets are dropped from the result of concatenation.  For instance, concatenating ($pk_1 \trcat pk_2 \trcat pk_3$) with ($pk_4 \trcat pk_5 \trcat pk_6$) results in ($pk_1 \trcat pk_2 \trcat pk_5 \trcat pk_6$) when
$pk_3 = pk_4$ and is undefined otherwise.  Hence, concatenating two 2-element traces leaves us with a 2-element trace. 
The \emph{concatenation} of two trace sets
$(S_1 \kand S_2)$ is defined as follows.
\begin{align*}
S_1 \kand S_2  = \{ & pk_{11}  \trcat  pk_{12}  \trcat  \cdots  \trcat  pk_{1(n-1)}  \trcat  pk_{22}  \trcat  \cdots  \trcat  pk_{2m} \mid \\
& pk_{11}  \trcat  \cdots  \trcat  pk_{1(n-1)}  \trcat  pk_{1n} \in S_1 , 
pk_{21}  \trcat  pk_{22}\cdots  \trcat  pk_{2m} \in S_2, 
pk_{1n} = pk_{21} \}
\end{align*}

NetKAT packet predicates ($\pred$) identify subsets of packets flowing through a network.  They include
simple tests of a packet field against a \emph{constant}
$(\field = \const)$ as well as any boolean combination of
such tests (conjunction is $\pand$; disjunction is $\por$; negation is $\pnot$; true is $\ptrue$; false is $\pfalse$).  We also use $(\Sigma_{a\in A} \,\pred(a))$ for an n-ary disjunction, and likewise
$(\Pi_{a\in A} \,\pred(a))$ for an n-ary conjunction.

Packet relations $PKR$ denote sets of packet pairs
(alternately, 2-element traces).  The relation
$\uppkr{f}{\const}$ denotes pairs $(\pk_1, \pk_2)$
such that $\pk_2$ is $\pk_1$ with field $f$ changed
to the constant $\const$.   When $f$ is the location field, assignment to $f$ represents movement of the packet from one device to another.  Other relations include
subsets of the identity relation ($\overline{\pred}$ -- every packet satisfying $\pred$ is paired with itself), cartesian product ($\pred_1 \crosspkr \pred_2$), composition of relations ($\PR_1 \comppkr \PR_2$),
union ($\PR_1 \orpkr \PR_2$), intersection ($\PR_1 \andpkr \PR_2$), and complement ($\notpkr$) of relations.

Finally, NetKAT expressions $K$ denote sets of traces.  They include packet relations, which generate 2-element traces, as well as the regular operators
concatenation ($\kand$), union ($\kor$), and star ($\ks{K}$).  Finally, \kdup~is the set of all 3-element traces of the form $\pk \ \pk\ \pk$.
Concatenation of two 2-element traces
results in a 2-element, but concatenation of a
2-element trace with a 3-element trace leaves
a 3-element trace.  Hence, $\kdup$~effectively
extends the length of a trace by one.
The following expression, which denotes a
set of 2-element traces, illustrates some of these
features.  

\begin{align*}
    (\lift{\dstip = 1.0.0.0} \ \kand \  &\uppkr{\locf}{A}) \  \kor  & \texttt{\% (1)}\\
    (\lift{\dstip = 1.0.0.1} \ \kand \  & \uppkr{\srcip}{2.0.0.0} \kand \uppkr{\locf}{B}) \  & \texttt{\% (2)} 
\end{align*}

\noindent
For any initial packet (starting
at any location) in the trace, 
 if its destination IP is 1.0.0.0, then
the second packet in the trace is the same as the first except its location field is $A$, and
(2) if its destination IP is 1.0.0.1, then the second packet's
source IP is 2.0.0.0 and its location is B.
%
In examples, we often use the following shorthand, which are easy
to define:  (1) $\havoc$ relates any two packets, (2) $\alltraces$ is the
set of all packet traces, and (3) $\alltraces(\pred)$ is the set of all
traces for which every packet satisfies $\pred$.

\paragraph*{Relational NetKAT}  Relational NetKAT~\cite{han26relational} extends NetKAT with
relations between sets of traces.  Such relations
may be viewed as network transformers.  In past work,
Relational NetKAT was used to specify intended \emph{changes} to
networks.  To compute the image of $K$ under the relation $R$ we write $K \rapp R$.  
Relations $R$ are built from primitives and
application of regular operators for concatenation ($\rconcat$),\footnote{For technical reasons,
concatenation on trace relations eliminates one
element, not two, as in concatenation of traces.
We use a different symbol for concatenation on
relations to highlight that difference. Relations
on one-element traces are admitted.}
 union ($\ror$),
and star ($\rs{R}$).
The key primitives are: 
\begin{itemize}
    \item $\rmap{\PR}{K}$, which relates each trace $\trace$ in $K$ to a
    new trace $\trace'$ generated from $\trace$ by applying the packet relation $\PR$ to all its elements;
    \item $\rfilter{\PR} \rconcat R$, which is like the relation $R$ except
    that elements $(\trace_1, \trace_2)$ of $R$ are
    discarded when the first packets in $\trace_1$ and $\trace_2$ do not satisfy
    the relation $\PR$;
    \item $\rinsert{K}$, which relates any two-element trace to all traces in $K$;     
    \item $\rdelete{K}$, which relates all traces in $K$ to any two-element trace; and
    \item $\rid{K}$, the identity relation, which relates all traces in $K$ to themselves. (The expression $\rid{K}$ can equivalently be written as $\rmap{\overline{1}}{K}$.
\end{itemize}

For example, to specify all paths through device $B$ should be
rerouted through device $C$ instead, we define
\(\textit{change\_path} = \rid{\alltraces} \rconcat \rdelete{\uppkr{\locf}{B} } \rconcat \rinsert{\uppkr{\locf}{C}} \rconcat \rid{\alltraces}.
\)
\noindent
Further, to change such paths while also ensuring that all paths not through $B$ are preserved, we define
$
\textit{full\_change} = \textit{change\_path} \ror \rid{\alltraces(\pnot (\locf = B))}
$.
Such a relation may be applied to any existing network to generate a changed
network:  $\textit{changed} = \textit{existing} \rapp \textit{full\_change}$.

\paragraph*{Weighted NetKAT}  Weighted NetKAT~\cite{acevedo26weight} extends NetKAT with a semi-ring of weights $w \in \mathcal{W}$, operations
for multiplication ($w \cdot w'$) and addition ($w \wor w'$), a
multiplicative identity $1$ and an additive identity $0$.   
A \emph{weighted trace} is a pair of a
trace and its weight $(\trace, w)$.  The denotation of a
weighted NetKAT expressions $W$, written $\dw{W}$,
is a set of weighted traces.  When we concatenate two
weighted traces, we combine their weights with multiplication.  When we take the union of two sets of weighted traces, we add the weights of traces that belong to both sets.  To modify weights
of a weighted NetKAT expression directly, one may use
the forms $W \wtimes w'$ or $w' \wtimes W$. When a weighted trace 
$(\trace, w)$ belongs to $\dw{W}$, 
$(\trace, w \cdot w')$ belongs to $\dw{W \wtimes w'}$, and
likewise, $(\trace, w' \cdot w)$ belongs to $\dw{w' \wtimes W}$. 

As an example, suppose that we are interested in the minimum latency between devices in a network.
We would use the tropical semiring of rational numbers (extended with $+\infty$ for unreachability) with $\cdot$
interpreted as rational addition, and $+$ as minimum.  
Hence, when concatenating traces, we add latencies (the latency of A followed by B is the latency of A plus the latency of B), and when unioning traces, we choose the shortest path, adopting the minimum of the latencies from either.  Hence the expression, 
\begin{align*}
(\lift{\locf = A} \kand \uppkr{\locf}{B} \wtimes 3 \kand \kdup \kand \uppkr{\locf}{C} \wtimes 5)\ & + 
(\lift{\locf = A} \kand \uppkr{\locf}{B} \wtimes 1 \kand \kdup \kand \uppkr{\locf}{C} \wtimes 1)
\end{align*}
\noindent
will denote paths through locations $A B C$ with latency 2.

\section{Parametric NetKAT by Example}

Parametric NetKAT extends the NetKAT family of languages by allowing parameters to
appear in the place of constants and by providing algorithms for finding the set of valuations for parameters that validate constraints.  For example, whereas NetKAT includes concrete tests 
such as $\dstip = 10.0.0.0$, Parametric NetKAT admits \emph{symbolic tests}
$\dstip = x$, and whereas NetKAT includes concrete field updates, such as
$\uppkr{\locf}{C}$, Parametric NetKAT admits \emph{symbolic updates} $\uppkr{\locf}{x}$.
Finally, to constrain parameters, Parametric NetKAT allows (boolean combinations of) conditions of the form $x = c$.  
%


\paragraph*{Methodology} To use Parametric NetKAT effectively, we suggest the following methodology.

\begin{enumerate}
    \item Parse context- and vendor-specific network data plane formats using existing NetKAT tools \cite{han25artifact} and produce a NetKAT expression $N$ that faithfully represents network semantics.
    For some applications, one must compare one data plane to a second data plane (or one data plane component to several other data plane components).  In such cases, one might parse and generate two (or more) network expressions $N_1$, $N_2$, \emph{etc.}
    \item Transform the expression $N$, focusing on the network subparts of interest, and inserting parameters where needed, using 
    a relation $RN$.  In other words, craft $RN$ and an expression of the form 
    $N \rapp RN$ to find the image of $N$ under relation $RN$.
    The relation $RN$ manipulates the network $N$ semantically, which is convenient: one does not have to understand the syntax of low-level
    device configurations, or the NetKAT encoding process chosen, to analyze network semantics.  The denotational semantics of Relational NetKAT defines the effect of the transformation precisely.
    \item Define constraints/queries ($Q$) over transformed expressions.  If $\pn$, $\pn_1$, $\pn_2$, etc are parametric
    NetKAT expressions (such as $N \rapp RN$) the primitive constraints include:
    \begin{itemize}
        \item $\pn = \emptyset$: Find the parameter valuations that make the Parametric NetKAT expression denote the empty set of traces,
        \item $\pn_1 = \pn_2$: Find the parameter valuations that make the
        denotations of $\pn_1$ and $\pn_2$ the same, and
        \item $\qs{f}{WN}$:  Find the parameter valuations
        such that the sum of the weights from all traces in $WN$ satisfies 
        the predicate $f : W \rightarrow Bool$.
    \end{itemize}
    \item Complete the query by combining parameter valuations from multiple constraints using standard operators for union ($Q_1 \cup Q_2$),
    intersection ($Q_1 \cap Q_2$) and complement ($\neg Q_1$).
\end{enumerate}
\noindent
The result of this process is a set of parameter valuations, represented compactly as a BDD, which users may further examine or materialize as needed.

The rest of this section presents a series of examples that craft queries for a variety of useful network diagnosis tasks.  Important take-aways from the examples include the range of different kinds of questions that may be posed, the compactness, modularity, and simplicity of the queries, and the interactions between features of the language. 

\paragraph{Example:  Which Packets?}  A long-standing network analysis question is simply \emph{which packets} can flow from point A to point B, through a way-point, or along other prescribed paths.  Moreover, one may ask what those packets look like at various points along the way, despite modifications by intermediate devices such as NATs. 

As an example, consider
a network with a firewall $FW$ designed to
control access to device $B$ from device $A$.   
We may want to ask for the set of packets that flow from $A$ to $B$, 
and more specifically, what the destination IPs of those packets are when they reach the firewall $FW$ somewhere in the middle of the path (we may want to use that information to add additional firewall rules if our access control policy is misimplemented).   We can construct such queries compositionally, in steps, by crafting a relation that transforms the given network, discarding portions irrelevant to our question, and using parameters to extract information of interest.  The first step in this process is to filter out traces that do not start at $A$ and end at $B$:
\[
R_{A,B} = \rfilter{\lift{\locf=A}} \rconcat \rid{alltraces} \rconcat \rfilter{\lift{\locf=B}}
\]
\noindent
Given the filtered set of traces, we could ask for the values of the destination IP addresses of packets that arrive at $FW$ along any path.  These values will be assigned to $x$.
\[
R_{dstIP} = \rid{alltraces} \rconcat \rfilter{\lift{\dstip=x\pand loc=FW}} \rconcat \rid{alltraces}
\]
\noindent
To complete the query, we apply the relations to the network encoding $N$ and ask for any values
of $x$ that generate a valid path:  $Q1 = (N \rapp R_{A,B} \rapp R_{dstIP} \not= \emptyset)$.
If we would like to gather more information, perhaps the source IP at $A$ (stored in $y$),
that is easy to do as well:
\[
\begin{array}{rcl}
R_{srcIP} & = & \rfilter{\lift{\srcip=y\pand loc=A}} \rconcat \rid{alltraces} \\
Q_2 & = & (N \rapp R_{A,B} \rapp R_{dstIP} \rapp R_{srcIP} \not= \emptyset)
\end{array}
\]
More generally, the combination of relational filtering and parameter placement
allows us to select any set of traces expressible in NetKAT, and then to extract any packet header field at any point in those paths.  If we extract multiple
headers, say headers $x_1, \ldots, x_n$, at possibly many points along the path, the system returns a set of tuples $(x_1, \ldots, x_n)$ --- each tuple is a valuation that makes a path possible. 

\paragraph*{Example: Multipath Differencing}\label{para:multi}
Many large networks have symmetric designs~\cite{vl2}.  Such symmetries simplify network construction, maintenance, and expansion.
When precise specifications for networks are unavailable, as is often the case, systems such as Campion~\cite{campion} and Batfish~\cite{batfish} have shown that one can uncover bugs by comparing two network components that should be 
similar to one another, and analyzing differences.

Inspired by Batfish~\cite{batfish}, we explore the 
\emph{multipath differencing problem}:  Given two paths from $A$ to $B$,
$\mathit{path}_1$ and $\mathit{path}_2$, we ask for the IP addresses
of packets treated differently along the paths.\footnote{In NetKAT, an expression that denotes
a set of traces, say $\{\textit{trace}_1,\textit{trace}_2\}$, where the
traces begin with the same packet, has often been 
interpreted as multicast.  Here, we interpret it as a model of the "available paths" in a multipath routing system~\cite{equal-cost-multipath-rfc}.  A given packet will only flow along one of the paths.  If a path fails, packets that would have be forwarded over the failed path are instead forwarded across one of the other available paths thanks to fast failover mechanisms~\cite{fast-reroute-rfc} and control plane follow-up.}  For example, we would like to
return IP $x$, when a packet with IP $x$ is dropped along 
$\mathit{path}_1$ but is forwarded along $\mathit{path}_2$.  
Such differences can reflect misconfigured firewalls or other errors along the way.  

To be more concrete, we consider all traces following the paths $ACB$ and $ADB$ respectively, defined as follows.  In our encoding, 
the cross product ($\locf= A \times  \locf=C$)
represents a hop from $A$ to $C$ with
packet headers varying in arbitrary ways
(in any way they might vary in the network).
The $\kdup$ expression "saves" the
packet at $C$ and the product
($1 \times \locf=B$) adds a final hop to $B$.
\[
\begin{array}{rcl}
\mathit{path_1} &=& (\locf= A \times  \locf=C) \kand \kdup \kand (1 \times \locf=B) \\
\mathit{path_2} &=& (\locf=A \times  \locf=D) \kand \kdup \kand (1 \times \locf=B)
\end{array}
\]
\noindent
Next, since we are intentionally
comparing two different paths, we wish to distinguish \emph{outcomes} along those paths rather than details of the trace en route.  To focus on outcomes, we construct a relation that collapses traces to just their endpoints.
In the following construction, $\rdelete{\alltraces}$,
eliminates an arbitrary-length input trace, while $\rinsert{havoc}$ inserts 
an arbitray 2-step trace.  To make sure the input and output traces match up
at their end points, as desired, we filter using the identity relation.
\[
\mathit{collapse}
=
\rfilter{\lift{1}}\, \rconcat
\rdelete{\alltraces}\, \rconcat
\rinsert{havoc}\, \rconcat
\rfilter{\lift{1}}
\]
\noindent
Finally, we construct a query that considers sets of IP addresses in groups according to their initial destination IP ($x$), constrains the paths followed to $\mathit{path}_1$ and $\mathit{path}_2$, collapses the traces to their endpoints, and checks for inequality.  The valuations of $x$ (a set of destination IP addresses) that exhibit differences along the two paths are potential implementation bugs that should be rectified. 
\[
\begin{array}{rcl}
    \mathit{fixIP}(i) & = & \rfilter{\lift{dst.ip=z}} \rconcat \rid{\mathit{path}_i} \\
    Q_3 & = & (N \rapp \mathit{fixIP}(1) \rapp \mathit{collapse}) \not= (N \rapp \mathit{fixIP}(2) \rapp \mathit{collapse})
\end{array}
\]

\paragraph*{Example:  Device Fault Tolerance}\label{para:device-fault}    To determine whether $A$ can reach $B$, despite a single failure at node $x$, we can ask
$Q_4 = (N \rapp R_{A,B} \rapp \mathit{available} = \emptyset)$ where
$\mathit{available}$ is $\rmap{\lift{\locf \not= x}}{\alltraces}$.
Recall that the packet relation $\lift{\locf \not= x}$ is the identity relation restricted to packets where $\locf \neq x$. 
Thus $\mathit{available}$ has the effect of discarding traces that pass through location $x$, while leaving traces that do not pass through $x$ unchanged.
Hence, the query collects values of $x$
for which there does not exist a path from $A$ to $B$, indicating that the network cannot tolerate failures at those nodes.  An example like this can be coded in plain (unparametric) NetKAT through a series of $N$ queries,
where $N$ is the number of nodes in the network, but (as we show in Section~\ref{sec:eval}) it is more efficient to do them all simultaneously in Parametric NetKAT.

To extend to any two failures,
$x$ and $y$, we use $\mathit{available}_2 = \rmap{\lift{\locf \not= x \pand \locf \not= y}}{\alltraces}$ in place
of $\mathit{available}$.
And if we are interested in reachability beyond fixed pairs of nodes $A$ and $B$, we can extract sources ($y$), sinks ($z$), and failures ($x$) that disconnect them using the following query.
\[
\begin{array}{rcl}
R_{yz}  & = & \rfilter{\lift{(\Sigma_{c\in \mathit{sources}} \, y = c) \pand \locf=y}} \rconcat \rid{\alltraces}) \rconcat 
\rfilter{\lift{(\Sigma_{c\in \mathit{sinks}} \, z = c) \pand \locf=z}}\\
Q_5 & = & (N \rapp R_{yz} \rapp \mathit{available} = \emptyset)
\end{array}
\]
\noindent 
Analyzing link fault tolerance rather than device fault tolerance is also possible---see Section \ref{sec:eval} for further information.

\OMIT{
\paragraph*{Example:  Link Fault Tolerance}\label{para:link-fault}
Another related question is whether a network can
tolerate some small, fixed number $N$ of link failures rather than device failures.
To encode such a question for $N=1$,
suppose we number the links in our network $0 \ldots E$
(let \textit{edges} be that set of numbers).
For each link $i$, \textit{src}(i) is the starting
location and $\textit{dst}(i)$ is the ending location of the link.  The question may be answered by introducing
a parameter $\textit{fail}$, which takes on a value
in \textit{edges} to indicate the failed link.  
The query $Q_6$, constructed below, finds
the values of $i$ that render $B$ unreachable from $A$.
\[
\begin{array}{rcl}
\textit{link}(i) & = & \locf = \textit{src}(i) \crosspkr \locf = \textit{dst}(i) \\
\textit{available\_link}(i) & = & \rfilter{\textit{fail} \not= i}  \rconcat \idname(link(i)) \\
\textit{available\_links} & = & \Sigma_{i\in edges} \textit{available\_link}(i) \\
\textit{available\_paths} & = & \ks{available\_links} \\
Q_6 & = & (N \rapp \textit{available\_paths} \rapp R_{A,B} = \emptyset)
\end{array}
\]

\zak{Do we need \textit{two} fault tolerance examples?  Is there a point that we can't make with one example?}
}

\paragraph{Example:  Fault Localization} 
A fault tolerance analysis tells us ahead-of-time that our network can tolerate certain failures.  Sometimes, however, network engineers will discover that their network \emph{has already failed} despite best efforts to engineer resilience.  In such a situation, the engineer needs to determine the root cause of the failure in order to fix it---this is the \emph{fault localization} problem,
which is one of a broader class of problems, known as \emph{network tomography} problems~\cite{tomography}, that involve inferring properties internal 
to a network from external observations.

More specifically, suppose we notice packets being dropped but we are not sure where.  By sending packets
along paths from sources to destinations, we can generate two
sets of observations: $\mathit{Reach}$ and $\mathit{Unreach}$ such that
for all $(a1,a2)$ in $\mathit{Reach}$, we know packets from $a1$ can reach $a2$, and
for all $(a1,a2)$ in $\mathit{Unreach}$, we know packets from $a1$ cannot currently reach $a2$. Our goal is to determine which single device failure $x$ might explain these observations.
Our encoding of the available paths borrows $\textit{available}$ from the fault tolerance example above.
This time, however, we generate a collection of many equations, one equation
for each observation, and we seek the list of possible device failures $x$ that explain all the observations at once. 
$R_{a1,a2}$ is all possible traces from $a1$ to $a2$ as usual.
\[
\begin{array}{rcl}
G_1 & = & \bigcap_{(a1,a2)\in \mathit{Reach}} (N  \rapp R_{a1,a2} \rapp \textit{available} \not= \emptyset) \\
G_2 & = & \bigcap_{(a1,a2)\in \mathit{Unreach}} (N \rapp R_{a1,a2}  \rapp \textit{available} = \emptyset) \\
Q_7 & = & G_1 \cap G_2
\end{array}
\]
\noindent
If there is no solution to the problem, then one could try modified queries checking for any 2 failed devices, or perhaps for failed links instead, to uncover conditions consistent with the observations. 

\paragraph*{Example: Multi-objective Synthesis}  A key difficulty in network management is the complexity of it all---one must satisfy different objectives for a multitude of different flows.\footnote{A flow is a set of related packets that follow the same paths through the network.}  A particular worry is that a
fix for one flow can unintentionally cause \emph{collateral damage}~\cite{Xieyang24Rela}, disrupting transmission (or access control) for some other flow that was being processed properly. 

Imagine $A$ transmits undesirable traffic to $B$ through the firewall $FW$ and we would like to learn $(\srcip, \dstip)$ pairs to block.  However, at the same time, $D$ is transmitting traffic to $E$ and that traffic
may also run through the firewall.  Our goal then is two-fold:
\begin{enumerate}
    \item[$G_1$:] Learn the $(\srcip, \dstip)$ pairs of packets
    that appear at $FW$ while in transit $A$-$B$ (so that blocking these pairs at $FW$ will drop all such traffic).
    \item[$G_2$:] Ensure the paths used by $D-E$ traffic are unaffected, whatever those paths may be.
\end{enumerate}
To achieve goal $G_1$, we can use a query similar to query $Q_2$ constructed earlier.  Below, we
ask for the $(x,y)$ pairs that allow us to block all traces from $A$ to $B$.
\[
\begin{array}{rcl}
R_{src,dst} & = & 
\rid{alltraces} \rconcat 
\rfilter{
\lift{\locf = FW \pand \srcip=x \pand \dstip=y}} \rconcat \rid{alltraces} \\
G_1 & = & (N \rapp R_{A,B} \rapp R_{src,dst} \not= \emptyset)
\end{array}
\]
To ensure the traffic from \(D\) to \(E\) is preserved, despite addition of a firewall rule at \(FW\), we construct the following constraints.
\[
\begin{array}{rcl}
R_{D,E} & = & \rfilter{\lift{\locf=D}} \rconcat \rid{alltraces} \rconcat \rfilter{\lift{\locf=E}} \\
R_{change} & = & 
\rid{alltraces}\rconcat
      \rfilter{\lift{\locf=FW \pand (\srcip \neq x \por \dstip \neq y)}}\rconcat
      \rid{alltraces}\,  + \\
      & & \rs{\rid{\locf \not= FW}} \\
\textit{old}_{DE} & = & N \rapp R_{D,E} \\
\textit{new}_{DE} & = & N \rapp R_{change} \rapp R_{D,E} \\    
G_2 & = & (\textit{old}_{DE} = \textit{new}_{DE})
\end{array}
\]
We satisfy both equations simulateously by
asking for the intersection of the solutions
to our equations:
$Q_8 = G_1 \cap G_2$.

\paragraph*{Example: Quantitative Reasoning}\label{para:q-reason}
Weighted traces can be used in conjunction with parameteres to enumerate valuations that exhibit certain quantitative properties.
One simple kind of quantitative reasoning deploys
an artic semi-ring to compute longest paths and constraints to detect 
overly long paths.  In this application, weights are naturals (along with $\infty$), concatentation of weights is addition, and the sum of weights is maximum.  The following
Weighted NetKAT expression, when treated as a relation, will compute 
longest paths, counting each hop in the path as one.
\[
\begin{array}{lcr}
W_{length} & = & (1\otimes havoc~dup)^*(1\otimes havoc) 
\end{array}
\]
When combined
with a thresholding function such as
$f = \lambda w. w > 4$, we
can detect paths with a length greater than 4, which may violate
network latency requirements.
More specifically, the following query identifies the initial location
and destination IP address of any flow that can follow a path longer than
length 4.
\[
\begin{array}{rcl}
R_{init} & = & \rfilter{\lift{\locf=x \pand \dstip=y}} \rconcat \rid{\alltraces} \\
Q_9 & = & \qs{f}{\wres{W_{length}}{(N \rapp R_{init})} }
\end{array}
\]

\paragraph*{Limitations}  While Parametric NetKAT expands the kinds of questions the NetKAT family of languages can answer, there remain useful network diagnosis questions that are beyond its reach.  For example, a Parametric NetKAT query can only generate valuations with a fixed number of elements.   It cannot extract a set of traces, for instance, where the traces may contain arbitrarily many packets.  Implementing such a query may be possible in the future, perhaps by using a NetKAT automaton as a representation of the set, but it is beyond the scope of this paper.  Likewise, the current system does not allow for symbolic weights and the extraction of a set of weights that satisfy some property.  Such an extension may be useful for implementing certain kinds of quantitative tomography problems.
\section{Syntax and Semantics}
In this section, we present the formal syntax and denotational semantics of our language.

\paragraph{Types of Expressions.}
As in NetKAT \cite{netkat,foster15coal}, we model a network as a packet-processing system, and view its behavior as the set of all packet traces that can arise within the system.  We enrich this model with parameters, which represent the choices that validate given constraints.

Formally, let \(\mathsf{Par}\) be a finite set of parameters and \(\mathsf{Fld}\) a finite set of packet fields. We 
assume that each field value and each parameter value ranges over a finite set \(\mathsf{Val}\subseteq \mathbb{N}\).
We then define
\[
V = \mathsf{Val}^{\mathsf{Par}},
\qquad
\mathit{Pk} = \mathsf{Val}^{\mathsf{Fld}},
\qquad
\mathit{Tr}(Pk)=\{\,pk_1\cdots pk_n \mid n\ge 2,\; pk_i\in \mathit{Pk}\,\}.
\]
Thus, a valuation \(v\in V\) assigns to each parameter \(x\in\mathsf{Par}\) a natural number \(v.x\in\mathsf{Val}\), and a packet \(pk\in\mathit{Pk}\) assigns to each field \(f\in\mathsf{Fld}\) a natural number drawn from a finite set of natural numbers \(\mathsf{Val}\subseteq \mathbb{N}\).
Equivalently, a packet may be written as a record of the form
\(
\{f_1=c_1;\;f_2=c_2;\;\dots;\;f_n=c_n\},
\)
where each \(c_i\in\mathsf{Val}\).
We interpret comparisons between variables and fields values in the standard way over \(\mathbb{N}\), and build Boolean expressions from these atomic predicates using \(0\), \(1\), \(\neg\), \(\land\), and \(\lor\).
Finally, \(\mathit{Tr}(Pk)\) denotes the set of finite sequences of packet of length at least \(2\); we write such traces as \(pk_1\cdots pk_n\), using juxtaposition for concatenation.

\paragraph{Weights and Semirings.}
Weighted NetKAT \cite{acevedo26weight} extends NetKAT with weights drawn from an \(\omega\)-continuous semiring \((S,+,\cdot,0,1)\).
An arbitrary semiring is not sufficient for our purposes, since the weighted semantics may involve countable sums.
Accordingly, the weight domain must carry an order \(\preceq\) such that \((S,\preceq)\) forms an \(\omega\)-complete partial order, \(0\) is the least element, the operations \(+\) and \(\cdot\) are \(\omega\)-continuous, and \(S\) admits countable sums.
For a detailed discussion of these assumptions, we refer the reader to the prior work \cite{acevedo26weight}.

In addition, we equip the semiring with a star operator \((-)^*\), and require that
\(
e^* = \sum\limits_{i\in\mathbb{N}} e^i
\). 
This operator is needed computationally: it enables our algorithms and implementation to perform the required closure operations efficiently.

\paragraph{Parametric NetKAT Syntax}
Parametric NetKAT extends NetKAT, Relational NetKAT and Weighted NetKAT
by allowing for the use of variables $x$ in place of constants ($c$).  
These variables can take on any valuation needed to satisfy specified constraints.  More specifically,
rather than limiting NetKAT expressions to testing fields against constants
(\(f=c\)), Parameterized NetKAT now allows fields to be tested against
variables (\(f=x\)), and rather than merely assigning constants to fields
(\(f\leftarrow c\)), Parameterized NetKAT now allows fields to be assigned 
variables (\(f\leftarrow x\)).  Finally, to constrain variables independently of
how they are used in tests or assignments, we allow equalities $x = c$ to appear in expressions.    

\OMIT{
Figure~\ref{fig:enkat-syntax} presents the syntax of Parameterized NetKAT, with grey boxes highlighting significant differences from past work~\cite{netkat,han26relational,acevedo26weight}.

Having introduced packets, valuations, and weights, we can now present the syntax of Parametric NetKAT, together with its extensions to the relational and weighted settings.
Most constructs in Figure~\ref{fig:enkat-syntax} follow prior work \cite{netkat,han26relational,acevedo26weight}, except for the two new forms
\[
f=x
\qquad\text{and}\qquad
f\leftarrow x,
\]
which are parameterized counterparts of the ordinary constant-based constructs \(f=c\) and \(f\leftarrow c\).
Intuitively, \(f=x\) tests whether the field \(f\) is equal to the value of the variable \(x\), while \(f\leftarrow x\) updates the field \(f\) to that value.}

In the Weighted NetKAT sublanguage, we deviate from past work in two ways. 
A new form (\(WN \wtimes w\)) multiplies a weighted expression by a weight on the right-hand side, symmetrically to (\(w \wtimes WN\)).
This construct does not appear in Weighted NetKAT \cite{acevedo26weight}, but it is useful for semirings whose multiplication is not commutative.
The second new form is (\(\wres{WN}{PN}\)), which restricts a Weighted NetKAT expression by assigning weight \(0\) to every trace \(\tau\) not appearing in \(PN\), while leaving the weights of traces in \(PN\) 
unchanged.  Alternatively, one may view 
(\(\wres{WN}{PN}\)) as applying the weights associated with traces in
$WN$ to the unweighted expression $PN$---this latter viewpoint reflects the most common way we use this form.  More specifically, our network parsing infrastructure generates an unweighted NetKAT expression $N$.  Then a programmer may write their own application-specific weighting function $WN$, 
within the weighted NetKAT sublanguage, and apply $WN$ to $N$ using restriction (\(\wres{WN}{N}\)). This design increases the
programmability and modularity of the system.

Finally, we introduce a \textit{query language} for parametric NetKAT expressions that allows us to combine multiple constraints generated by all three sublanguages.  A query is interpreted as a subset of $\power{V}$, representing the set of all valuations under which a given formula holds.
There are three classes of atomic queries: emptiness checking \(PN=\emptyset\), which generates the set of valuations \(v\) under which \(PN\) is empty;
equivalence checking \(PN_1=PN_2\), which generates the set of valuations \(v\) under which \(PN_1\) and \(PN_2\) are equivalent; and
weighted queries \(\qs{f}{WN}\), which take a Boolean-valued function \(f: W \to \textsf{Bool}\) together with a weighted expression \(WN\), and generates the set of valuations $v$ such that the aggregate weight of all traces of \(WN\) satisfies \(f\).
Figure~\ref{fig:enkat-syntax} presents
the complete syntax of the system, highlighting extensions 
of past work in grey.

Aside from these changes, the syntax of 
the NetKAT, Parameterized NetKAT and Weighted NetKAT are largely unchanged from past work, though the three sublanguages have now been combined into
one and complex queries generated from boolean combinations of emptiness checking, equivalence checking, and constrained weight aggregation now available for the first time. 

\begin{figure}[t]
\vspace{-5pt}
\[
\begin{array}{lll}
\multicolumn{3}{l}{\textbf{Syntax:}} \\
\val & ::= & \const~|~\colorbox{gray!30}{$x$} \\
Pred &::=& \pfalse~|~\ptrue~|~ \colorbox{gray!30}{$x=c$}~|~\colorbox{gray!30}{$f=\val$}~|~Pred_1 \por Pred_2~|~Pred_1 \pand Pred_2~|~\pnot Pred\\
PkR &::=& \empkr~|~\idpkr~|~\colorbox{gray!30}{\uppkr{f}{\val}}|
Pred\times Pred~|PkR_1\comppkr PkR_2~|PkR_1 \orpkr PkR_2~|PkR_1\andpkr PkR_2~|\neg PkR\\
PN & ::= & PkR~|~PN\triangleright RN~|~ PN_1+PN_2 ~|~PN_1\kand PN_2~|~PN_1\kcap PN_2~|~PN_1\kdiff PN_2~|~\ks{PN}~|~dup\\
RN &::= &\rfilter{PkR} ~|~ \rmap{PkR}{PN} ~|~ \rdelete{PN} ~|~ \rinsert{PN} ~|~ RN_1 \rconcat RN_2 ~|~ RN_1 \ror RN_2 ~|~ \rs{RN} \\
WN & ::= & PN~|~WN\wtimes w~|~w\wtimes WN~|~\colorbox{gray!30}{$\wres{WN}{PN}$}~| WN_1+WN_2 ~|~WN_1\kand WN_2~|~\ks{WN}\\
\colorbox{gray!30}{$Q$} & ::= & 
\colorbox{gray!30}{$PN=\emptyset~|~PN_1=PN_2~|~\qs{f}{WN}~|~Q_1\cap Q_2 ~|~ Q_1 \cup Q_2 ~|~\neg Q~$}
\end{array}
\]
\[\begin{array}{lll@{\qquad\qquad}lll@{\qquad}}
\multicolumn{6}{l}{\textbf{Common Abbreviations:}} \\
\havoc & = & \ptrue \crosspkr \ptrue
& \alltraces(\pred) & = & \ks{(\pred \times 1 \kand \kdup )} \kand\pred \times \pred  \\
\lift{\pred} & = & \idpkr \andpkr (\pred \crosspkr \pred) 
& \alltraces & = & \alltraces(\ptrue) \hfill\\
\end{array}
\]
\vspace{-13pt}
\caption{Parametric NetKAT Syntax.  Extensions relative to past work highlighted in grey.}
    \vspace{-15pt}
\label{fig:enkat-syntax}
\end{figure}

\begin{figure}[t]
\centering
\vspace{-5pt}
\begin{subfigure}[t]{0.95\linewidth}
\centering
\[
\begin{aligned}
\dpred{\pfalse}(v) &= \varnothing \\
\dpred{\ptrue}(v)  &= \mathit{Pk} \\
\dpred{f = x}(v)
  &= \{\, pk \in \mathit{Pk} \mid pk.f = v.x \,\} \\
\dpred{f = c}(v)
  &= \{\, pk \in \mathit{Pk} \mid pk.f = c \,\} \\
\dpred{x = c}(v)
  &= \{\, pk \in \mathit{Pk} \mid v.x = c \,\} \\
\dpred{\pnot Pred}(v)
  &= \mathit{Pk} \setminus \dpred{Pred}(v) \\
\dpred{Pred_1 \por Pred_2}(v)
  &= \dpred{Pred_1}(v) \cup \dpred{Pred_2}(v) \\
\dpred{Pred_1 \pand Pred_2}(v)
  &= \dpred{Pred_1}(v) \cap \dpred{Pred_2}(v)
\end{aligned}
\]
\caption{Denotational semantics of predicates.}
\label{fig:pred-semantics}
\end{subfigure}

\begin{subfigure}[t]{0.95\linewidth}
\centering
\[
\begin{aligned}
\dpkr{\empkr}(v) &= \varnothing \\
\dpkr{\idpkr}(v)
  &= \{\, (pk,pk) \mid pk \in \mathit{Pk} \,\} \\
\dpkr{\uppkr{f}{x}}(v)
  &= \{\, (pk_1,pk_2) \mid pk_2 = pk_1[f \leftarrow v.x] \,\} \\
\dpkr{\uppkr{f}{c}}(v)
  &= \{\, (pk_1,pk_2) \mid pk_2 = pk_1[f \leftarrow c] \,\} \\
\dpkr{Pred_1 \times Pred_2}(v)
  &= \dpred{Pred_1}(v) \times \dpred{Pred_2}(v) \\
\dpkr{PkR_1 \comppkr PkR_2}(v)
  &= \left\{
     \begin{aligned}
     (pk_1,pk_3)\ \mid \ &\exists pk_2 \in \mathit{Pk}, \\
     &(pk_1,pk_2) \in \dpkr{PkR_1}(v), \\
     &(pk_2,pk_3) \in \dpkr{PkR_2}(v)
     \end{aligned}
     \right\} \\
\dpkr{PkR_1 \orpkr PkR_2}(v)
  &= \dpkr{PkR_1}(v) \cup \dpkr{PkR_2}(v) \\
\dpkr{PkR_1 \andpkr PkR_2}(v)
  &= \dpkr{PkR_1}(v) \cap \dpkr{PkR_2}(v) \\
\dpkr{\neg PkR}(v)
  &= (\mathit{Pk} \times \mathit{Pk}) \setminus \dpkr{PkR}(v)
\end{aligned}
\]
\caption{Denotational semantics of packet relations.}
\label{fig:pkr-semantics}
\end{subfigure}
\vspace{-10pt}
\caption{Denotational semantics of predicates and packet relations.}\vspace{-10pt}
\label{fig:pred-pkr-semantics}
\end{figure}

\begin{figure}[t]
\vspace{-5pt}
\begin{align*}
\dpn{PkR}(v)
  &= \{\, pk_1 pk_2 \mid (pk_1, pk_2) \in \dpkr{PkR}(v) \,\} \\
\dpn{PN\triangleright RN}(v)
  &= \{\, \tau_2\mid \tau_1\in\dpn{PN}(v), (\tau_1,\tau_2)\in\dpr{RN}(v),|\tau_2|\ge 2  \,\} \\
\dpn{PN_1 + PN_2}(v)
  &= \dpn{PN_1}(v) \cup \dpn{PN_2}(v) \\
\dpn{PN_1 \kcap PN_2}(v)
  &= \dpn{PN_1}(v) \cap \dpn{PN_2}(v) \\
\dpn{PN_1 \kdiff PN_2}(v)
  &= \dpn{PN_1}(v) \setminus \dpn{PN_2}(v) \\
\dpn{PN_1 \kand PN_2}(v)
  &= \{\, \tau \mid \exists \tau_1 \in \dpn{PN_1}(v),\ \exists \tau_2 \in \dpn{PN_2}(v),\ \tau = \tau_1 \kand \tau_2 \,\} \\
\dpn{\ks{PN}}(v)
  &= \bigcup_{n \ge 0} \dpn{PN^n}(v),
  \qquad
  \text{where }
  PN^0 = \idpkr,\;
  PN^{n+1} = PN^n \kand PN \\
\dpn{dup}(v)
  &= \{\, pkpkpk \mid pk \in \mathit{Pk} \,\}\\\\
    \dpr{\rfilter{PkR}}(v) & = \{ (pk_1, pk_2) \mid (pk_1, pk_2) \in \dpkr{PkR}(v) \} \\
    \dpr{\rmap{PkR}{PN}}(v)& = \{ (pk_1 \trcat \cdots \trcat pk_n, pk_1' \trcat \cdots \trcat pk_n') \mid 
    pk_1 \trcat \cdots \trcat pk_n \in \dpn{PN}(v), \\
    & \hspace{5.6cm} \forall i \in [1,n]. (pk_i, pk_i') \in \dpkr{PkR}(v) \} \\
    \dpr{\rdelete{PN}}(v) & = \{ (pk_1 \trcat \cdots \trcat pk_n, pk) \mid pk \in Pk, pk_1 \trcat \cdots \trcat pk_n \in \dpn{PN}(v) \} \\
    \dpr{\rinsert{PN}}(v) & = \{ (pk, pk_1 \trcat \cdots \trcat pk_n) \mid pk \in Pk,  pk_1 \trcat \cdots \trcat pk_n \in \dpn{PN}(v) \} \\
    \dpr{RN_1+RN_2}(v)& = \dpr{RN_1}(v) \cup \dpr{RN_2}(v) \\
    \dpr{RN_1\rconcat RN_2}(v)& = \left\{
        \begin{aligned}
        &(pk_1 \trcat \cdots \trcat pk_{n1} \trcat \cdots \trcat pk_{n2},\; pk_1' \trcat \cdots \trcat pk_{m1}' \trcat \cdots \trcat pk_{m2}') \mid \\
        &\quad (pk_1 \trcat \cdots \trcat pk_{n1},\; pk_1' \trcat \cdots \trcat pk_{m1}') \in \dpr{RN_1}(v), \\
        &\quad (pk_{n1} \trcat \cdots \trcat pk_{n2},\; pk_{m1}' \trcat \cdots \trcat pk_{m2}') \in \dpr{RN_2}(v)
        \end{aligned}
    \right\} \\
    \dpr{RN^*}(v)& = \bigcup_{n \geq 0} \dpr{RN^n}(v), \quad \text{where } R^0 = \rfilter{\havoc},\quad RN^{n+1} = RN^n \rconcat RN
    \\\\
        \dw{PN}(v) & = \{(\tau,1) \mid \tau\in \dpn{PN}(v)\}\cup\{(\tau,0) \mid \tau\notin \dpn{PN}(v)\} \\
    \dw{\wres{WN}{PN}}(v) & = \{(\tau,w) \mid \tau\in \dpn{PN}(v)\land (\tau,w)\in \dw{WN}(v)\}\cup\{(\tau,0) \mid \tau\notin \dpn{PN}(v)\} \\
    \dw{w\wtimes WN}(v) & = \{(\tau,w\cdot w') \mid (\tau,w')\in \dw{WN}(v)\} \\
    \dw{WN\wtimes w}(v) & = \{(\tau,w'\cdot w) \mid (\tau,w')\in \dw{WN}(v)\} \\
    \dw{WN_1\wor WN_2}(v)& = \{\, (\tau,w_1+w_2)\mid (\tau_1,w_1) \in \dw{WN_1}(v),\  (\tau,w_2) \in \dw{WN_2}(v)\} \\
    \dw{WN_1\wand WN_2}(v)& = \{\, (\tau,\sum\limits_{(\tau_1,w_1) \in \dw{WN_1}(v),\  (\tau_2,w_2) \in \dw{WN_2}(v), \tau = \tau_1 \kand \tau_2} w_1\cdot w_2)\} \\
    \dw{WN^*}(v)& = \{(\tau,\sum\limits_{(\tau,w)\in \dw{W^n}(v)}w)\}, \quad \text{where } WN^0 = \lift{1},\quad WN^{n+1} = WN^n \wand WN
\end{align*}
    \caption{Semantics of NetKAT, Relational NetKAT, and Weighted NetKAT.
    }\vspace{-10pt}
    \label{fig:selected-sem}
\end{figure}

\begin{figure}[t]
\vspace{-5pt}
\centering
\begin{align*}
\dq{PN=\emptyset} &= \{\,v \mid \dpn{PN}(v)=\emptyset\,\}, \\
\dq{PN_1=PN_2} &= \{\,v \mid \dpn{PN_1}(v)=\dpn{PN_2}(v)\,\}, \\
\dq{\qs{f}{WN}} &= \left\{\,v \mid f\!\left(\sum_{(\tau,w)\in \dwn{WN}} w\right)\right\}, \\
\dq{Q_1\cap Q_2} &= \dq{Q_1}\cap\dq{Q_2}, \\
\dq{Q_1\cup Q_2} &= \dq{Q_1}\cup\dq{Q_2}, \\
\dq{\neg Q_1} &= V\setminus \dq{Q_1}.
\end{align*}
\caption{Denotational semantics of queries.}\vspace{-10pt}
\label{fig:query-semantics}
\end{figure}

\paragraph{Denotational Semantics.}
A Parametric NetKAT expression denotes a function from 
valuations \(v\in V\) to network behaviors, while a Parametric Query $Q$ denotes a set of valuations---those valuations that satisfy the given constraints.  The types of each denotation function follow.
\[
\dpred{\cdot} : Pred \to V \to \power{\mathit{Pk}},
\qquad
\dpkr{\cdot} : PkR \to V \to \power{\mathit{Pk} \times \mathit{Pk}},
\]
\[
\dpn{\cdot} : PN \to V \to \power{\mathit{Tr}(Pk)},
\qquad
\drn{\cdot} : RN \to V \to \power{\mathit{Tr}(Pk)\times \mathit{Tr}(Pk)},
\]
\[
\dwn{\cdot} : WN \to V \to \power{\mathit{Tr}(Pk)\times W)}
\qquad
\dq{\cdot}: Q \to \power{V}
.
\]
The definitions of these functions appear in Figures~\ref{fig:pred-pkr-semantics} and~\ref{fig:selected-sem}.  The semantics is a conservative extension of NetKAT:  If a subexpression 
does not mention a parameter then it behaves uniformly across all valuations,
and its semantics coincides with the traditional semantics of NetKAT.
For example,
\(
\dpred{f=c}(v)=\{\, pk \in \mathit{Pk} \mid pk.f=c \,\}.
\)
In contrast, the semantics of (\(f=x\)) depends on the valuation \(v\): a packet satisfies this test exactly when its field \(f\) matches the value assigned to \(x\) by \(v\).
Hence,
\(
\dpred{f=x}(v)=\{\, pk \in \mathit{Pk} \mid pk.f=v.x \,\}.
\)

The semantics of \(Q\) is straightforward once the semantics of \(PN\), \(RN\), and \(WN\) have been defined. 
For the basic queries \(PN=\emptyset\), \(PN_1=PN_2\), and \(\qs{f}{WN}\), the semantics of \(Q\) simply collects all valuations satisfying the corresponding constraint.
The composite queries \(Q_1\cap Q_2\), \(Q_1\cup Q_2\), and \(\neg Q\) are then interpreted by taking the intersection, union, and complement of these sets of valuations, respectively.

\section{Automata and Symbolic Analysis}\label{sec:aut}

In this section, we provide algorithms for calculating the denotational 
semantics of any query $Q$.  
The key observation is that parameters can be treated as additional packet fields---a similar observation as was used to implement symbolic control plane analysis in Expresso~\cite{expresso} and NV~\cite{nv}. 
In other words, a Parametric NetKAT program over packet space \(Pk\) and valuation space \(V\) can be simulated by an ordinary NetKAT-style semantics over the enlarged packet space \(Pk\times V\).
Under this simulation, we can compile Parametric NetKAT and its extensions into their non-parameterized automata models, thereby obtaining both a compatibility towards existing NetKAT features and a compilation method.
We present the automata we use in Section~\ref{subsec:aut} and the embedding theorem in Section~\ref{subsec:emb}.
Once this simulation framework is in place, we develop new algorithms for emptiness checking, equivalence queries, and weighted queries over the translated automata.

\subsection{Automata}\label{subsec:aut}
Our work involves three kinds of automata. The definitions of
NetKAT automata and Relational NetKAT automata are drawn directly from work by Xu~\cite{han26relational}.  Likewise, definitions for weighted NetKAT automata follow from work by Acevedo~\cite{acevedo26weight}, with minor notational changes. Definitions of NetKAT and weighted NetKAT automata are presented in the following;  Relational NetKAT automata 
are relegated to the Appendix for space reasons.


\begin{definition}
A \textbf{NetKAT automaton} is a tuple $M=(S,S_0,S_f,\Delta)$,
where \(S\) is a finite set of states, \(S_0\subseteq S\) is the set of initial states, \(S_f\subseteq S\) is the set of accepting states, and
$\Delta : S\times S \to 2^{Pk\times Pk}$
is a transition relation.
\end{definition}
The key difference between ordinary automata and NetKAT automata lies in
the transition relation.
In an ordinary automaton, a transition depends only on the current input symbol (the type of the transition function is typically $S\times S \to 2^{\Sigma}$), whereas in a NetKAT automaton, a transition depends on both the current input packet and the packet produced at the previous step.  Said another way,
the packet processed during the next step is related to the packet processed by the current step.  The output of the transition function (a relation between current and past packets) exhibits that difference. 

The semantics of a NetKAT automaton \(M\) is defined inductively over packet traces with the help of  a labeled transition system of the following form.
\(
(s_0,x_0)\stackrel{\tau}{\longrightarrow}(s_n,x_n)
\)
Such a transition states that in starting state \(s_0\) with input packet \(x_0\), the automaton can process the trace \(\tau\) and reach state \(s_n\) with output packet \(x_n\).  Legal transitions are defined as follows.
\begin{itemize}
    \item \textbf{Base case:} $(s_0, x_0) \stackrel{\epsilon}{\longrightarrow} (s_0, x_0)$
    \item \textbf{Inductive case:} If $(x_0, x_1) \in \Delta(s_0, s_1)$ and $(s_1, x_1) \stackrel{\tau}{\longrightarrow} (s_n, x_n)$, then
    $
    (s_0, x_0) \stackrel{x_1 \tau}{\longrightarrow} (s_n, x_n).
    $
\end{itemize}
\noindent
We define \emph{language accepted by \(M \)} to be $L(M)$, which is the set of traces that begin with an initial state and terminate in an accepting state.  Formally:
\(
L(M) \defeq \{ pk_0\tau \mid \exists s_0\in S_0,  s_f \in S_f, pk_0 \in \mathit{Pk}. (s_0,pk_0) \stackrel{\tau}{\longrightarrow} (s_f, pk_n) \} 
\)

\paragraph{Weighted NetKAT Automata.}
In the original Weighted NetKAT work \cite{acevedo26weight}, the authors propose a weighted automaton model in which the transition function maps each state in \(S\) to a monadic structure.
Here, we adopt an equivalent but simpler presentation, using an ordinary weighted transition relation of the form
\(
S\times S \to 2^{\mathit{Pk}\times \mathit{Pk}\times W},
\)
so as to better align with our presentations of NetKAT automata and Relational NetKAT automata.
This weighted automaton model is essentially the same as that of prior work \cite{acevedo26weight}, differing only in notation.

\begin{definition}
A \textbf{Weighted NetKAT automaton} is a tuple
\(
M = (S,I,F,\Delta),
\)
where \(S\) is a finite set of states, \(I : S \to W\) assigns an initial weight to each state, \(F : S \to W\) assigns a final weight to each state, and
\(
\Delta : S \times S \to 2^{\mathit{Pk}\times \mathit{Pk}\times W}
\)
is a transition relation.
\end{definition}

The transition semantics of \(M\) is again defined inductively over packet traces. In this case, a labeled transition
\(
(s_0,x_0)\eqby{\tau}{w}{\longrightarrow}(s_n,x_n)
\)
means that, starting from state \(s_0\) with input packet \(x_0\), the automaton can process the trace \(\tau\), reach state \(s_n\), produce output packet \(x_n\), and accumulate weight \(w\). Note that the weight is obtained by summing over all possible intermediate transitions.

\begin{itemize}
    \item \textbf{Base case:}
    \(
    (s_0,x_0)\eqby{\epsilon}{1}{\longrightarrow}(s_0,x_0).
    \)

    \item \textbf{Inductive case:}
    If
    \(
    (x_0,x_1,w_1)\in \Delta(s_0,s_1)
    \) and \(
    (s_1,x_1)\eqby{\tau}{w_2}{\longrightarrow}(s_n,x_n),
    \)
    then
    \(
    (s_0,x_0)\eqby{x_1\tau}{w_1w_2}{\longrightarrow}(s_n,x_n).
    \)
\end{itemize}

The language accepted by \(M\) consists of all input traces together with their accumulated weights, starting from an initial state and ending in an accepting state:
\[
L(M)
=
\{
(x_0\tau,w)
\;|\;
w =
\sum_{(s_0,x_0)\eqby{\tau}{w'}{\longrightarrow}(s_n,x_n)}
I(s_0)\,w'\,F(s_n)
\}.
\]

\subsection{Embedding}\label{subsec:emb}

\OMIT{
Many NetKAT extensions go well beyond the constructs of original NetKAT \cite{netkat}.
For example, the constructs \(PN_1 \cap PN_2\) and \(PN_1 \backslash PN_2\) were introduced in \cite{Moeller24Katch}, \(PkR\) and \(PN \triangleright RN\) were introduced in \cite{han26relational}, and Weighted NetKAT was studied in \cite{acevedo26weight}.
Some of these constructs are highly non-standard and technically involved.
This raises a natural question: can we integrate parameterization into the NetKAT ecosystem in a principled and lightweight way?
}

Implementing and optimizing new regular languages and their automata is challenging.
Therefore, rather than developing entirely new machinery for Parametric NetKAT, we aim to reuse the existing NetKAT machinery whenever possible.

In our semantics, every parametrized language is interpreted functionally: it maps each parameter valuation in \(V\) to a corresponding non-parametric semantics.
For example, a predicate expression \(Pred\) has semantics of type
\(
V \to \power{Pk},
\)
which describes the packet semantics under each concrete valuation \(v\).
Although this functional view is natural for expressing the role of parameters, it is less convenient for compilation and algorithmic purposes.
In particular, the valuation space \(V\) may be very large---for example, it may contain as many as \(2^{32}\) values when parameterizing an IP address---while our goal is often to identify all valuations \(v\) satisfying a given constraint.
Instead of applying valuations one by one, we seek a symbolic treatment of parameters.
A natural idea, then, is to transform the functional semantics
\(
V \to \power{Pk}
\)
into an equivalent set-based semantics
\(
\power{Pk\times V},
\)
which records each valuation \(v\) together with its associated non-parametric behavior.
For instance, the predicates \(f=x\) and \(f=c\), which are originally interpreted as
\[
\lambda v.\ \{\, pk \mid pk.f = v.x \,\}
\qquad\text{and}\qquad
\lambda v.\ \{\, pk \mid pk.f = c \,\},
\]
are translated into the set-based semantics
\[
\{\, (pk,v) \mid pk.f = v.x \,\}
\qquad\text{and}\qquad
\{\, (pk,v) \mid pk.f = c \,\}.
\]
\noindent
Similarly, the semantics of \(PN\), \(RN\), and \(WN\) can be transformed into the set-based semantics
\[
\power{Tr(Pk\times V)},
\qquad
\power{Tr(Pk\times V)\times Tr(Pk\times V)},
\qquad\text{and}\qquad
\power{Tr(Pk\times V)\times W},
\]
respectively, as we will formalize later in this section.

The next observation is that this translated set-based semantics has exactly the same type as ordinary non-parametric NetKAT over the extended packet space \(Pk\times V\).
Specifically, in the non-parametric setting, NetKAT, Relational NetKAT, and Weighted NetKAT have semantics of type
\[
\power{Tr(Pk)},
\qquad
\power{Tr(Pk)\times Tr(Pk)},
\qquad\text{and}\qquad
\power{Tr(Pk)\times W},
\]
respectively.
Moreover, the way we access the value of a parameter variable \(x\) under a valuation \(v\), namely \(v.x\), is exactly analogous to the way we access the value of a packet field \(f\) from a packet \(pk\), namely \(pk.f\).
This suggests a natural way to reuse the NetKAT ecosystem: encode parameter variables as additional field names, and encode valuations as part of the extended packet space.

A naive way to proceed would be to translate Parametric NetKAT directly into ordinary NetKAT syntax.
However, while such a translation is theoretically possible, it leads to an exponential blowup.
For example, consider the Parametric NetKAT construct \(f=x\) over the field space \(Fld=\{f\}\).
A naive translation into ordinary NetKAT over the extended field space \(Fld=\{f,x\}\) would be
\(
\sum_{v\in V} (f=v)\cdot(x=v).
\)
The underlying reason is that ordinary NetKAT supports comparisons and assignments only between packet fields and constants, but not directly between packet fields and other packet fields.
Thus, a direct syntactic translation would require explicit enumeration of all valuations, which is clearly undesirable.

The key observation is that this explosion can be avoided if we compile not to NetKAT syntax, but directly to NetKAT automata.
Unlike the surface language, NetKAT automata support arbitrary packet relations of type \(2^{Pk\times Pk}\) as transitions, and therefore also support relations of type \(2^{(Pk\times V)\times(Pk\times V)}\) over the extended packet space.
Consequently, we do not need exponentially large syntactic encodings for constructs such as \(f=x\) and \(f\leftarrow x\); instead, we can compile them directly using their translated set semantics in the automaton.
Once this encoding is in place, the existing automata constructions for the various NetKAT extensions can be reused without modification.

In the remainder of this section, we first formalize the transformation from functional semantics to set-based semantics.
We then show how to systematically reuse existing NetKAT automata constructions to obtain a correct compilation procedure.
At this point, although we can compile the translated set semantics of Parametric NetKAT, it is no longer presented in its original functional form.
Accordingly, the algorithms for the translated automata must also be revised, both to account for the new semantics and to answer queries that require more than yes-or-no answers.
This will be the subject of Section~\ref{subsec:alg}.

\paragraph{Set-based Interpretation.}
To show that parameters can be treated as extra packet fields concretely, we define a semantic translation from the parametric NetKAT over packet space $Pk$ and parameter space $V$ to a semantic over $Pk\times V$. We call this translation the \emph{Set-based interpretation}.
Previously, we defined the semantic functions
\(
\dpred{\cdot}, \dpkr{\cdot}, \dpn{\cdot}, \drn{\cdot}, \text{and}\ \dwn{\cdot},
\)
which process objects in \(Pred\), \(PkR\), \(PN\), \(RN\), and \(WN\), respectively, and return a function from valuations \(v \in V\) to the corresponding denotation under \(v\).
We now reinterpret these valuation-indexed semantics as ordinary semantics over the product packet space.
\[
\begin{aligned}
\ipred{\cdot} &: (V \to \power{Pk}) \to \power{Pk \times V} \\
\ipred{f}
  &= \{\, (pk,v) \mid pk \in f(v) \,\},
\\[0.8ex]
\ipkr{\cdot} &: (V \to  \power{Pk \times Pk}) \to \mathcal P((Pk \times V)\times(Pk \times V)) \\
\ipkr{f}
  &= \{\, ((pk_1,v),(pk_2,v)) \mid (pk_1,pk_2) \in f(v) \,\},
\\[0.8ex]
\ipn{\cdot} &: (V \to \mathcal P(Tr(Pk))) \to \mathcal P(Tr(Pk \times V)) \\
\ipn{f}
  &= \{\, ((pk_1,v)(pk_2,v)\cdots(pk_n,v)) \mid pk_1pk_2\cdots pk_n \in f(v) \,\},
\\[0.8ex]
\irn{\cdot} &: (V \to \mathcal P(Tr(Pk)\times Tr(Pk))) \to \mathcal P(Tr(Pk \times V)\times Tr(Pk \times V)) \\
\irn{f}
  &= \left\{
     \bigl((pk_{11},v)(pk_{12},v)\cdots(pk_{1n_1},v),\;
           (pk_{21},v)(pk_{22},v)\cdots(pk_{2n_2},v)\bigr)
     \;\middle|\; \right. \\
  &\qquad \left.
     (pk_{11}pk_{12}\cdots pk_{1n_1},\;
      pk_{21}pk_{22}\cdots pk_{2n_2}) \in f(v)
     \right\},
\\[0.8ex]
\iwn{\cdot} &: (V \to \mathcal P(Tr(Pk)\times W)) \to \mathcal P(Tr(Pk \times V)\times W) \\
\iwn{f}
  &= \{\, ((pk_1,v)(pk_2,v)\cdots(pk_n,v),w) \mid  (pk_1pk_2\cdots pk_n,w) \in f(v) \,\}.
\end{aligned}\]
We promote the packet space to \(Pk\times V\) by attaching the same valuation \(v\) uniformly to every packet in a trace.
For example, if the semantics of a Parametric NetKAT expression \(PN\) under valuation \(v\) contains the trace \(pk_1pk_2\cdots pk_n\), then the corresponding set-based semantics contains the trace
\(
(pk_1,v)(pk_2,v)\cdots(pk_n,v).
\) Although this extension of the packet space is simple, it faithfully reflects the original functional semantics: once the valuation \(v\) is fixed, it remains unchanged throughout the entire trace.
This property will also be crucial for algorithm design, since it allows us to treat the parameter valuation as fixed along a run, as we will see later in Section~\ref{subsec:alg}.

Once the set-based semantics is in place, we can then show, essentially immediately, that this interpretation is homomorphic with respect to the original denotational semantics.
In turn, this homomorphism guarantees that the existing correctness results for semantic-based automata constructions in NetKAT continue to apply in the parametric setting.

\begin{theorem}[Homomorphism]\label{thm:homo}
The Set-based interpretation is homomorphic with respect to the semantic constructors of Parametric NetKAT.

\paragraph{Parametric NetKAT}
For all valuations \(v\in V\), the following hold:
\begin{align*}
\ipn{\dpn{PkR}}
  &= \{\, (pk_1,v)(pk_2,v) \mid ((pk_1,v),(pk_2,v)) \in \ipkr{\dpkr{PkR}} \,\}, \\
\ipn{\dpn{PN \triangleright RN}}
  &= \{\, \tau_2 \mid \tau_1 \in \ipn{\dpn{PN}},\ (\tau_1,\tau_2)\in \irn{\dpr{RN}},\ |\tau_2|\ge 2 \,\}, \\
\ipn{\dpn{PN_1 + PN_2}}
  &= \ipn{\dpn{PN_1}} \cup \ipn{\dpn{PN_2}}, \\
\ipn{\dpn{PN_1 \kcap PN_2}}
  &= \ipn{\dpn{PN_1}} \cap \ipn{\dpn{PN_2}}, \\
\ipn{\dpn{PN_1 \kdiff PN_2}}
  &= \ipn{\dpn{PN_1}} \setminus \ipn{\dpn{PN_2}}, \\
\ipn{\dpn{PN_1 \kand PN_2}}
  &= \{\, \tau \mid \exists \tau_1 \in \ipn{\dpn{PN_1}},\ \exists \tau_2 \in \ipn{\dpn{PN_2}},\ \tau = \tau_1 \kand \tau_2 \,\}, \\
\ipn{\dpn{\ks{PN}}}
  &= \bigcup_{n\ge 0} \ipn{\dpn{PN^n}},
  \qquad \text{where } PN^0 = \idpkr,\quad PN^{n+1}=PN^n \kand PN, \\
\ipn{\dpn{dup}}
  &= \{\, (pk,v)(pk,v)(pk,v) \mid pk\in \mathit{Pk},\ v\in V \,\}.
\end{align*}

\end{theorem}
\begin{proof}
By direct unfolding of the definitions.
We defer the full proof to the Appendix.
\end{proof}

Similarly, one can prove corresponding homomorphism theorems for Relational NetKAT and Weighted NetKAT; we defer these to the Appendix.
An immediate consequence of these homomorphism results is that we can reuse the automata-construction techniques developed in prior work
\cite{Moeller24Katch,acevedo26weight,han26relational,prob-netkat}.

\begin{theorem}[Correctness of Compilation]
\label{thm:correct-compile}
For every \(PN\), \(RN\), and \(WN\), we can construct a NetKAT automaton \(M\), a transducer \(T\), and a Weighted NetKAT automaton \(WM\) such that
\[
L(M)=\ipn{\dpn{PN}},
\qquad
L(T)=\irn{\dpr{RN}},
\qquad
L(WM)=\iwn{\dw{WN}}.
\]
\end{theorem}

\begin{proof}
By Theorem~\ref{thm:homo}, every automata construction already proved correct for ordinary NetKAT, Relational NetKAT, and Weighted NetKAT also applies to their parametric interpretations.
For example, let \(PN_1\) and \(PN_2\) be Parametric NetKAT expressions where we want to compile the expression \(PN_1+PN_2\), and suppose we have NetKAT automata \(M_1\) and \(M_2\) such that
\[
L(M_1)=\ipn{\dpn{PN_1}}
\qquad\text{and}\qquad
L(M_2)=\ipn{\dpn{PN_2}}.
\]
Prior work gives a construction of a NetKAT automaton \(M_3\) such that
\(
L(M_3)=L(M_1)\cup L(M_2).
\)
Therefore,
\(
L(M_3)=\ipn{\dpn{PN_1}}\cup \ipn{\dpn{PN_2}}
      =\ipn{\dpn{PN_1+PN_2}},
\)
which automatically gives us the desired automaton.

Thus, it remains only to handle the constructs that are new in our setting.
For the atomic construct \(PkR\) which contains the new parameterized expression such as \(f=x\) and \(f\leftarrow x\), we can construct an automaton with two states and a single transition to represent its semantics; we defer the construction to the appendix.
For the construct \(WN \wtimes w\), compilation is immediate by symmetry with the existing construction for \(w \wtimes WN\).
For the construct \(\wres{WN}{PN}\), compilation is also straightforward: we take the product of the automaton for \(PN\) with the Weighted NetKAT automaton for \(WN\), and assign weight \(0\) to traces rejected by \(PN\).
Therefore, every construct of \(PN\), \(RN\), and \(WN\) admits a correct automata compilation, yielding the result.
\end{proof}

\subsection{Algorithms}\label{subsec:alg}
\begin{algorithm}[t]
\caption{Inferring valuations for which a translated NetKAT automaton is non-empty}
\label{alg:emptiness}
\textbf{Input:} A NetKAT automaton \(M=(S,S_0,S_f,\Delta)\) over packet space \(Pk\times V\), such that \(L(M)=\ipn{\dpn{PN}}\).

\textbf{Output:} The set of parameter valuations
\[
\{\, v \mid \dpn{PN}(v)\neq\emptyset \,\}.
\]

\begin{enumerate}
    \item Initialize \(Reach : S \to 2^{Pk\times V}\) and \colorbox{gray!30}{\(V_{\mathit{sol}}\subseteq V\)} by
    \[
    Reach(s)=
    \begin{cases}
    Pk\times V & \text{if } s\in S_0,\\
    \emptyset & \text{otherwise,}
    \end{cases}
    \qquad
    \colorbox{gray!30}{$V_{\mathit{sol}}\gets \emptyset$}.
    \]

    \item \textbf{While} some \(Reach(s)\) changes:
    \begin{enumerate}
        \item For each changed \(s\in S_f\),
        \[
        \colorbox{gray!30}{$V_{\mathit{sol}} \gets V_{\mathit{sol}} \cup \{\, v \mid (pk,v)\in Reach(s) \,\}$}.
        \]

        \item For each changed \(s\in S\), and each \(s'\in S\),
        \[
        \begin{aligned}
        Reach(s') \gets\;& Reach(s') \cup \{(pk_2,v)\mid \colorbox{gray!30}{$v\notin V_{\mathit{sol}}$}\land \\
        &((pk_1,v),(pk_2,v))\in \Delta(s,s') \land (pk_1,v)\in Reach(s)\}.
        \end{aligned}
        \]
    \end{enumerate}

    \item Return \(V_{\mathit{sol}}\).
\end{enumerate}
\end{algorithm}

In this section, we present algorithms for evaluating all queries in our language.
Recall that the semantics of every query is a set of parameter valuations satisfying the corresponding constraint.
For composite queries such as \(Q_1\cap Q_2\), \(Q_1\cup Q_2\), and \(\neg Q\), computation is straightforward: the desired set of valuations is obtained by set intersection, union, and complement, respectively.
Thus, the main technical task is to handle the atomic queries \(PN=\emptyset\), \(PN_1=PN_2\), and \(\qs{f}{WN}\).

At the algorithmic level, there are two essential new procedures in this section: an emptiness-checking algorithm and an aggregation algorithm.
The emptiness-checking algorithm computes all parameter valuations under which the translated NetKAT automaton reaches a final state.
It therefore directly answers the query \(PN=\emptyset\), and it also yields an algorithm for \(PN_1=PN_2\) by reducing equivalence to emptiness via the symmetric difference
\(
PN_1\backslash PN_2 + PN_2\backslash PN_1 = \emptyset.
\)

The aggregation algorithm is designed for weighted queries of the form \(\qs{f}{WN}\).
Such a query asks for all valuations \(v\) such that, under valuation \(v\), the total weight of all traces satisfies the Boolean predicate \(f\).
Thus, the main task is first to compute, for each valuation \(v\), the aggregate weight of all traces associated with \(v\).
This is exactly the purpose of our aggregation algorithm.
Once this aggregate has been computed, the query result is obtained simply by filtering with \(f\).

Fortunately, we do not need to design these procedures from scratch.
Both the emptiness-checking algorithm and the aggregation algorithm can be obtained by adapting existing algorithms for non-parametric automata.
In what follows, we present these constructions and highlight the modifications and optimizations needed to make them work in the parametric setting.

\begin{algorithm}[t]
\caption{Aggregate sum of weights over all traces of a Weighted NetKAT automaton}
\label{alg:aggregate}
\textbf{Input:} A Weighted NetKAT automaton
\(WM=(S,I,F,\Delta)\) over \(Pk\times V\), such that \(L(WM)=\iwn{\dw{WN}}\).

\textbf{Output:} The set associating each valuation with its aggregated weight:
\[
\{\, (v,\sum_{(\tau,w')\in \dw{WN}(v)} w') \mid v\in V \,\}.
\]

\begin{enumerate}
    \item
    As shown in prior work \cite{acevedo26weight}, the operations \(+\), \(\cdot\), and \((-)^*\) on transitions in Weighted NetKAT automaton are well defined, and can be efficiently computed from the corresponding \(+\), \(\cdot\), and \((-)^*\) operations of the underlying weight semiring \(W\).

    Based on these operations, apply the standard automata state-elimination algorithm \cite{Brzozowski1963SignalFG} until all intermediate weighted transitions have been eliminated and only a single final transition remains.

    \item Let
    \[
    \mathit{state\_elimination}((pk_1,v_1),(pk_2,v_2))
    \]
    denote the weight on the final transition from input packet \((pk_1,v_1)\) to output packet \((pk_2,v_2)\).
    Return
    \[
    \colorbox{gray!30}{$
    \left\{\,
    \left(v,\sum_{pk_1,pk_2\in Pk} \mathit{state\_elimination}((pk_1,v),(pk_2,v))\right)
    \;\middle|\;
    v\in V
    \right\}.
    $}
    \]
\end{enumerate}
\end{algorithm}

\paragraph{Emptiness Checking.}
The emptiness-checking algorithm (See Algorithm~\ref{alg:emptiness}) maintains two data structures: (1) \(Reach(s)\), which records the currently reachable packets at each state \(s\); and
(2) \(V_{\mathit{sol}}\), which records the valuations that have already reached a final state.
The algorithm repeatedly propagates reachable packets until a fixed point is reached, and then returns \(V_{\mathit{sol}}\).

Compared with the non-parametric emptiness-checking algorithm for NetKAT automata, the main addition is the solution set \(V_{\mathit{sol}}\), highlighted in the grey boxes.
This set serves two purposes.
First, it records the valuations satisfying the query, thereby providing more information than a simple yes-or-no answer.
Second, it enables an early-exit optimization: once a valuation \(v\) has already been shown to reach a final state, there is no need to continue propagating transitions for that valuation.
This is reflected in Step~(2)(b), where we update \(Reach(s')\) only using valuations \(v\notin V_{\mathit{sol}}\).
The soundness of this optimization follows from Theorem~\ref{thm:correct-compile}, which gives
\(
L(M)=\ipn{\dpn{PN}}.
\)
This translated semantics guarantees that the valuation component \(v\) remains unchanged throughout all transitions.
Hence, once a valuation has been identified at a final state, it can be safely discarded from further exploration at any other state.

With this intuition in place, the correctness statement is straightforward.

\begin{theorem}
Algorithm~\ref{alg:emptiness} (Emptiness Checking) is correct with respect to its output specification.
\end{theorem}

\begin{proof}
See the Appendix.
\end{proof}

\paragraph{Weighted Queries.}
The role of the aggregation algorithm is to compute, for each parameter valuation \(v\), the total weight \(w\) associated with \(v\).
To implement such an algorithm, we use of the classic state-elimination algorithm for automata \cite{Brzozowski1963SignalFG}.

The most common use of state elimination is to convert an automaton into an equivalent regular expression.
For Weighted NetKAT automata, however, we can reinterpret \(+\), \(\cdot\), and \((-)^*\) as the combination, concatenation, and closure operations on weighted transitions, as defined in prior work \cite{acevedo26weight}.
Under this interpretation, the state-elimination algorithm still applies.
Since each weighted transition has type
\(
2^{(Pk\times V)\times (Pk\times V)\times W},
\)
the final transition produced by state elimination has the same type.
It therefore represents, for each input-output packet pair, the weight obtained by summing over all traces connecting that input packet to that output packet.

The remaining step involves post-processing the final transition.
Because it already aggregates weights with respect to the initial and final packets, all that remains is to sum over packets while grouping by the parameter valuation.
By Theorem~\ref{thm:correct-compile}, we have
\(
L(WM)=\iwn{\dw{WN}},
\)
which guarantees that the valuation component \(v\) remains unchanged throughout every transition.
Hence, by summing over all input and output packets carrying the same valuation \(v\), we obtain exactly the desired aggregate weight for \(v\).

\begin{theorem}
Algorithm~\ref{alg:aggregate} (Aggregate-Sum) is correct with respect to its output specification.
\end{theorem}

\begin{proof}
See the appendix.
\end{proof}
\section{Evaluation} \label{sec:eval}
Having developed the theory of Parametric NetKAT, we now turn to its practical performance.
We first describe the implementation, and then evaluate it on real-world network topologies from Topology Zoo \cite{TopologyZoo}, as well as on large industrial benchmarks drawn from Alibaba's internal network traffic flows \cite{rela-arxiv} and AWS cloud network configurations \cite{batfish}.

\subsection{Implementation}
Our OCaml implementation is based on a substantial reworking of both the Relational NetKAT implementation \cite{han26relational} and the \href{https://ocaml.org/u/2a8a23e46502abca0f4c2ee71a24d904/mlbdd/0.7.3/doc/MLBDD/index.html}{MLBDD} library.
The original Relational NetKAT codebase contains roughly \(3000\) lines of code, while our extensions required a deep rewriting of its core components together with a substantial modification of MLBDD (2000 LoC), promoting its underlying representation from Binary Decision Diagrams (BDDs) to Algebraic Decision Diagrams (ADDs).
In addition, we implemented roughly \(4000\) lines of new features and tests, for a total of approximately \(9000\) lines of code. We now describe the key components of the implementation.

\paragraph{BDD and BDD layout.}
Binary Decision Diagrams (BDDs) \cite{BryantBdd} are a classical symbolic data structure for compactly representing large Boolean objects.
In networking, BDDs and their variants, such as FDDs and NDDs, have been widely used in a variety of network-analysis tools \cite{hsa,han26relational,Moeller24Katch,NDD}.
In our work, we use BDDs as the underlying symbolic representation because they provide greater flexibility in variable ordering, which is crucial for avoiding the exponential blowup introduced by parameterization.

Concretely, although Section~\ref{subsec:emb} shows how to compile a Parametric NetKAT program over \(Pk\) into a NetKAT automaton over \(Pk\times V\) by treating the valuation space \(V\) as additional packet fields, extra care is needed to optimize performance when representing these simulated valuation fields. 

Packet fields are often \(16\)- or \(32\)-bit values.
For example, suppose \(f\) is a \(32\)-bit packet field representing an IPv4 address, and let $f.0,f.1,\dots,f.31$ denote its bit-level encoding.
Now, consider the very common Parametric NetKAT expression \(f=x\) and it's bitwise encoding:
\(
\forall i\in [0,31].\ f.i=x.i.
\)
Representing this common property under an unfavorable BDD variable ordering
becomes prohibitively expensive.  For instance, if the variables are ordered as
\[
f.0,f.1,f.2,\dots,f.31,x.0,x.1,x.2,\dots,x.31,
\]
then the BDD contains \(3\times 2^{32}-1\) nodes, which is far too large to be practical.  In contrast, if we interleave the bits as
\[
f.0,x.0,f.1,x.1,f.2,x.2,\dots,f.31,x.31,
\]
then the resulting BDD has only \(98\) nodes, which is entirely manageable.

After choosing BDDs as the core data structure, we determine the BDD variable ordering according to this interleaving principle.
In previous NetKAT implementations \cite{Smolka15compile,Moeller24Katch,han26relational}, the BDD/FDD layout is sequential: for two fields \(f\) and \(g\), the bits of the two fields are not interleaved, but instead arranged as
\(
f.0,f.1,f.2,\dots,f.31,g.0,g.1,g.2,\dots,g.31.
\)
In our implementation, by contrast, we provide the flexibility to generate arbitrary BDD layouts---sequential, interleaving, or hybrid---via the function \texttt{make\_layout} in \texttt{Eval.ml}.
Users may interleave any collection of fields, keep them sequential as in ordinary NetKAT, or combine the two approaches in a hybrid layout.
On top of this, we introduce an abstraction layer based on the data structure \texttt{field\_layout} and function \texttt{generate\_layout} in \texttt{Eval.ml}.
At this level, users need only specify the dependency of each variable on packet fields---that is, which fields are related to a variable through constructs such as \(f=x\) or \(f\leftarrow x\)---together with the bit-width of each field.
From this information, our infrastructure automatically generates an appropriate BDD variable layout.

\paragraph{ADD and Matrix Closure.}
Algebraic Decision Diagrams (ADDs) \cite{ADD} are a classical extension of BDDs in which the leaf values are drawn from an arbitrary semiring rather than the Boolean domain \(\{\mathit{true},\mathit{false}\}\).
Although several public libraries---such as \href{https://www.cs.uleth.ca/~rice/cudd_docs/}{CUDD}, \href{https://meddly.sourceforge.io/index.html}{Meddly}, \href{https://add-lib.scce.info/}{ADD-Lib}, and \href{https://trolando.github.io/sylvan/}{Sylvan}---implement ADDs (also known as MTBDDs), to the best of our knowledge none provides direct support for fully user-defined semirings.
We therefore extended the \href{https://ocaml.org/u/2a8a23e46502abca0f4c2ee71a24d904/mlbdd/0.7.3/doc/MLBDD/index.html}{MLBDD} library to obtain the first ADD library with direct support for user-defined semiring weights.

After implementing ADDs, the next step is to use them to represent Weighted NetKAT programs.
As discussed in Section~\ref{sec:aut} and in prior work \cite{acevedo26weight}, each transition relation of a weighted automaton can be represented as a weighted matrix indexed by \(Pk\times Pk\).
Prior work \cite{ADD} on ADDs shows that matrix operations such as addition and multiplication can be implemented efficiently with ADD operations.
Thus, the main missing ingredient is matrix closure, that is, the Kleene star operation on weighted matrices.
For this purpose, we adopt the divide-and-conquer matrix closure method of prior work \cite{MatIte} and translate it into an ADD-based algorithm. We put the full details of this translated algorithm in the Appendix.

\paragraph{Weighted NetKAT Implementation.}
As a by-product of our development, we also implemented the first Weighted NetKAT system.
After extending our symbolic backend from BDDs to ADDs, we are able to lift the existing BDD-based NetKAT compilation pipeline to an ADD-based compilation pipeline for Weighted NetKAT.
Most of the compilation pipeline follows the method described in prior work \cite{acevedo26weight}, except that our implementation performs on-the-fly automaton construction using derivative-based techniques \cite{anti96dev}.

\subsection{Benchmarks}

Our evaluation is designed to validate that Parametric NetKAT can solve packet enumeration problems efficiently, even at the scale of large industrial networks.
We evaluate our approach on three benchmark suites drawn from industrial or widely used network artifacts: Topology Zoo \cite{TopologyZoo}, Rela \cite{rela-arxiv}, and Batfish \cite{batfish}.
These benchmarks cover scenarios including fault tolerance, longest paths, network synthesis, and packet inspection queries.

\paragraph{Topology Zoo Tests.}
\begin{figure*}[t]\vspace{-5pt}
\centering

\begin{subfigure}[t]{0.48\textwidth}
\centering
\includegraphics[width=\textwidth]{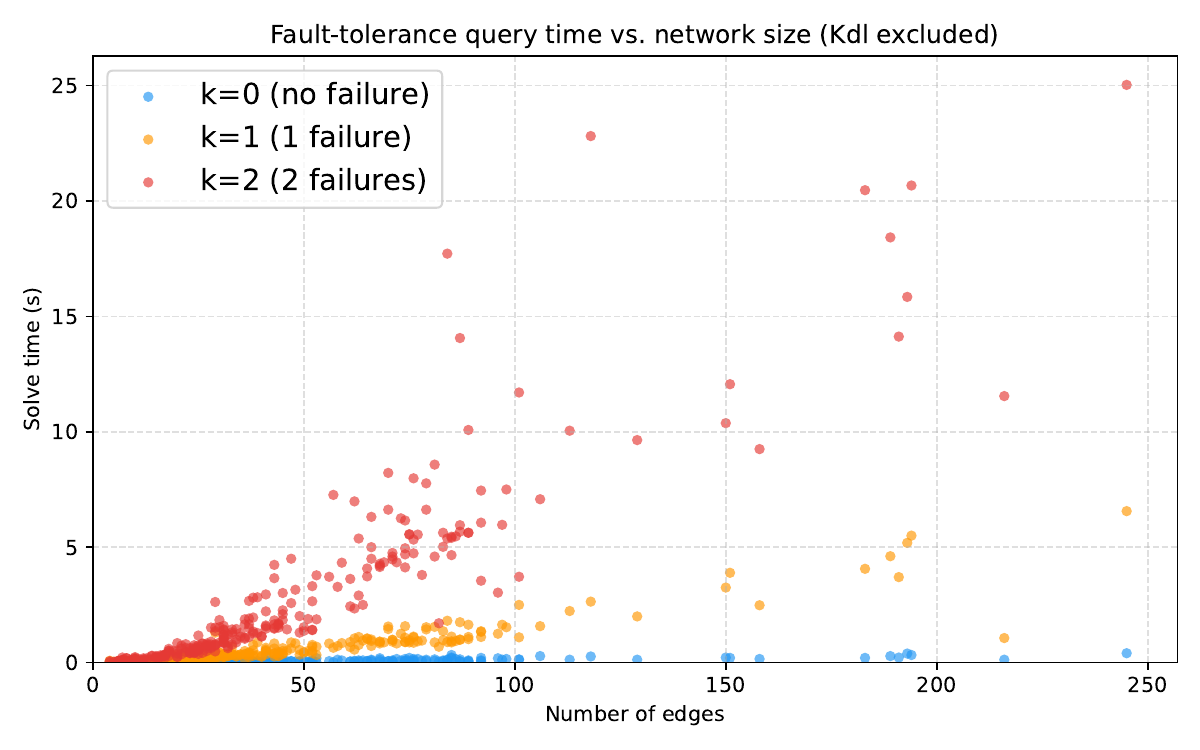}
\caption{Full reachability: running time (seconds) versus network size (excluding \texttt{Kdl.gml}).}
\label{fig:topozoo-full-scatter}
\end{subfigure}
\hfill
\begin{subfigure}[t]{0.48\textwidth}
\centering
\includegraphics[width=\textwidth]{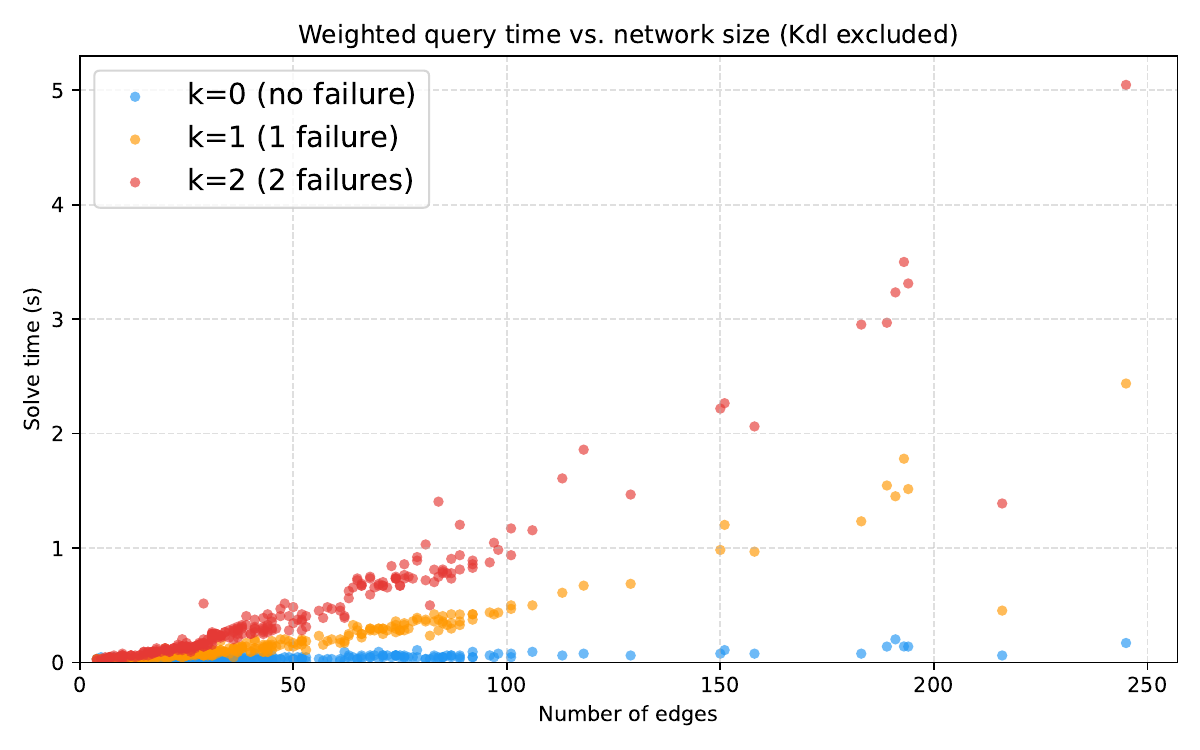}
\caption{Longest path: running time (seconds) versus network size (excluding \texttt{Kdl.gml}).}
\label{fig:topozoo-shortest-scatter}
\end{subfigure}

\vspace{0.8em}

\begin{subtable}[t]{0.95\textwidth}
\centering
\setlength{\arrayrulewidth}{0.8pt}
\begin{tabular}{|r|c|c|c|c|c|c|c|c|}
\hline
& \multicolumn{4}{c|}{Full reachability} & \multicolumn{4}{c|}{Longest path} \\
\cline{2-9}
\multicolumn{1}{|c|}{$k$}
& avg.\ (s) & median (s) & Kdl (s) & speedup
& avg.\ (s) & median (s) & Kdl (s) & speedup \\
\hline
0 & 0.08 & 0.03 & 4.69   & 1.00
  & 0.05 & 0.03 & 3.31   & 1.00 \\
1 & 0.83 & 0.36 & 47.09  & 5.66
  & 0.31 & 0.11 & 28.22  & 11.23 \\
2 & 3.38 & 1.41 & 131.34 & 101.76
  & 0.60 & 0.25 & 38.34  & 266.56 \\
\hline
\end{tabular}
\caption{Summary statistics for Topology Zoo benchmarks under \(k\)-link failures.}\vspace{-5pt}
\label{tab:topozoo-merged}
\end{subtable}
\caption{Topology Zoo evaluation results. 
}\vspace{-10pt}
\label{fig:topozoo-all}
\end{figure*}

The Internet Topology Zoo \cite{TopologyZoo} is an open dataset consisting of real-world telecommunications and data-network maps.
It contains \(261\) network topologies, ranging from small networks with only a few devices to large networks, such as the KDL case, with \(754\) devices and \(899\) links.
This dataset has been widely used to evaluate a variety of network-analysis tools in different scenarios \cite{Moeller24Katch,nv,SMTTopology}.

In our experiments, we use all the topologies in the dataset to test the scalability of our approach.
We evaluate link fault tolerance on two kinds of tasks:
\begin{enumerate}
    \item whether full reachability, that is, reachability between every pair of devices, is preserved under arbitrary \(0\)-, \(1\)-, or \(2\)-link failures; and
    \item whether the longest path between two devices remains below a given threshold under arbitrary \(0\)-, \(1\)-, or \(2\)-link failures.
\end{enumerate}
For task~(1), we introduce parameters on each link to indicate whether that link has failed.
For each link \(i\), let \(\textit{src}(i)\) and \(\textit{dst}(i)\) denote its source and destination locations, respectively.
We use \(\textit{fail}_j\) to denote the \(j\)-th failed link.
A link \(i\) may be traversed only when none of the parameters \(\textit{fail}_j\) is equal to \(i\).
Accordingly, we define the relation for inserting \(k\)-link failures as
\(
\textit{k\_failure}
=
\left(
\sum_{i\in edges}
\rfilter{\lift{\prod_{1\le j\le k} \textit{fail}_j \neq i}}
\idname(\locf = \textit{src}(i) \crosspkr \locf = \textit{dst}(i))
\right)^*.
\)

Once this \(k\)-link-failure insertion is defined, the two tasks are straightforward to express.
For task~(1), we want to check reachability between \emph{arbitrary} pairs of devices in the network.
To do so, we reuse the \texttt{collapse} relation from \nameref{para:multi}, which maps each trace to its initial and final packet pair.
Thus, we check whether
\(
N \rapp \textit{k\_failure} \triangleright \texttt{collapse}
=
N \triangleright \texttt{collapse},
\)
which expresses that the full reachability relation is preserved under every allowed failure scenario.

For task~(2), we choose two random devices \(A\) and \(B\), together with a random integer \(n\) in the range \(0\) to \(|devices|\).
We then evaluate whether the longest path is smaller than \(n\), using the longest-path query introduced in Paragraph~\nameref{para:q-reason}, combined with the \(k\)-link-failure insertion:
\(
\qs{f}{\wres{W_{length}}{(N \rapp R_{A,B} \rapp \textit{k\_failure})}}.
\)

The experimental results are summarized in Figure~\ref{fig:topozoo-all}.
Figures~\ref{fig:topozoo-full-scatter} and~\ref{fig:topozoo-shortest-scatter} plot running time against network size, measured by the number of edges, for all Topology Zoo benchmarks except \texttt{kdl.gml}.
We use the number of edges as the size measure because our fault-tolerance queries parameterize the network per edge.
We exclude \texttt{kdl.gml} from the scatter plots because it is a clear outlier: it contains \(899\) edges, which is more than three times as many as the second-largest network in the dataset.
Its running time is reported separately in Table~\ref{tab:topozoo-merged}.

As the figure shows, the average and median running times are both below \(5\) seconds across all test settings, indicating that the approach is efficient in practice.
Moreover, even on a network with roughly \(900\) edges, checking \(2\)-link fault tolerance takes only about \(2\) minutes, which further demonstrates the scalability of our method.

An additional quantity of interest is the \emph{speedup} column.
This column compares the running time of a single parameterized NetKAT query with the time required to enumerate all possible link-failure scenarios sequentially.
For \(k=0\), we set the speedup to \(1\) by definition.
For \(k=1\) and \(k=2\), we compute the speedup as
\(
\frac{T(k=0)\cdot |edges|}{T(k=1)}
\)
 and 
\(
\frac{T(k=0)\cdot |edges|\cdot (|edges|-1)}{2\cdot T(k=2)},
\)
respectively. This shows that Parametric NetKAT can greatly accelerate packet-enumeration tasks by answering with a single symbolic query, rather than enumerating the cases one by one.

\paragraph{Rela Tests.}

\begin{table*}[t]\vspace{-5pt}
\centering
\setlength{\arrayrulewidth}{0.8pt}
\begin{tabular}{|l|cc|cc|cc|cc|}
\hline
& \multicolumn{2}{c|}{$k=0$} & \multicolumn{2}{c|}{$k=1$} & \multicolumn{2}{c|}{$k=2$} & \multicolumn{2}{c|}{$k=3$} \\
\cline{2-9}
Benchmark
& avg (s) & max (s)
& avg (s) & max (s)
& avg (s) & max (s)
& avg (s) & max (s) \\
\hline
Identity Empty
& 0.09 & 0.20
& 0.12 & 0.27
& 0.15 & 0.39
& 0.17 & 0.41 \\
\hline
Delete Empty
& 0.10 & 0.39
& 0.14 & 0.63
& 0.24 & 1.80
& 0.72 & 8.08 \\
\hline
Reroute Empty
& 0.07 & 0.17
& 0.13 & 1.22
& 0.27 & 6.12
& 0.42 & 20.52 \\
\hline
Identity Weight
& 0.19 & 0.44
& 0.26 & 0.58
& 0.32 & 0.77
& 0.37 & 0.95 \\
\hline
Delete Weight
& 0.13 & 0.55
& 0.18 & 1.58
& 0.28 & 8.36
& 0.91 & 54.67 \\
\hline
Reroute Weight
& 0.12 & 0.30
& 0.17 & 0.59
& 0.28 & 6.67
& 0.44 & 25.55 \\
\hline
Identity Equiv.
& 0.19 & 0.44
& 0.25 & 0.66
& 0.29 & 0.78
& 0.35 & 1.08 \\
\hline
Delete Equiv.
& 0.22 & 0.48
& 0.40 & 1.19
& 1.52 & 8.22
& 10.24 & 133.99 \\
\hline
Reroute Equiv.
& 0.20 & 0.47
& 0.78 & 5.06
& 5.69 & 133.72
& - & $\ge$1000 \\
\hline
\end{tabular}
\caption{Running times for the parametric benchmark suite. All values are in seconds over 500 instances. For Reroute Equiv., the $k=3$ case is omitted because it occasionally timed out after 1000 seconds.}\vspace{-15pt}
\label{tab:param-bench-full}
\end{table*}

The second dataset we use is the Rela benchmark suite \cite{rela-arxiv}.
Rela is derived from Alibaba's internal network and contains \(2460\) devices.
Analyzing such a large network directly would be infeasible, so the benchmark partitions the network into \(21112\) traffic flows, each of which typically traverses \(30\)–\(40\) devices.
Moreover, for each traffic-flow instance, the dataset provides both a pre-update network and a post-update network, which we compile into NetKAT programs \(N_{pre}\) and \(N_{post}\), respectively.

We use this dataset to evaluate \emph{network synthesis} tasks.
Starting from the pre-update network \(N_{pre}\), we apply a parameterized update to synthesize \(N'_{pre}\), and then ask which parameter choices yield a desired network property.
Such properties include making two devices unreachable, ensuring that the longest path between two devices is below a threshold, or transforming the network so that it matches the intended post-update network \(N_{post}\).

In this experiment, we evaluate the scalability of all three algorithms---emptiness, equivalence, and weighted queries---under three update scenarios, each corresponding to a different way of transforming \(N_{pre}\) into \(N'_{pre}\).
In these tests, we parameterize location information for all devices in a traffic flow; encoding the \(30\)–\(40\) devices in a typical flow usually requires \(5\)–\(6\) bits.
In addition, for each algorithm we vary a parameter \(k\), which denotes the number of parameters introduced into the synthesis task.

Concretely, we evaluate the three algorithms on the following synthesis tasks:
\begin{enumerate}
    \item \textbf{Emptiness:} whether some choice of parameters makes two randomly chosen devices \(A\) and \(B\) unreachable, i.e.,
    \(
    N_{pre}' \triangleright R_{A,B} = \emptyset.
    \)

    \item \textbf{Weighted query:} whether some choice of parameters ensures that the longest path between two devices is below a threshold: 
    \(
    \qs{\lambda w.\, w<n}{\wres{W_{length}}{(N_{pre}'\triangleright R_{A,B})}}.
    \)

    \item \textbf{Equivalence:} whether some choice of parameters makes the updated network equivalent to the intended post-update network, i.e.,
    \(
    N_{pre}' = N_{post}.
    \)
\end{enumerate}

We consider the following three update scenarios:
\begin{itemize}
    \item \textbf{Identity baseline.}
    We make no change and set \(N'_{pre}=N_{pre}\).
    In this case, the parameter \(k\) simply counts dummy parameters introduced into the BDD data structure. We use this as a baseline for comparing the running times of the three kinds of queries, as well as for measuring the overhead introduced by allocating new parameter fields.

    \item \textbf{Device Deletion.}
    In \nameref{para:device-fault}, we show that one may wish to check whether the network remains functional under \(k\)-device failures, and to analyze its behavior under such failures.
    Here, we delete \(k\) devices and define
    \(
    N'_{pre} =
    N_{pre}\triangleright
    \rmap{\lift{loc\neq x_1 \cdot \cdots \cdot loc\neq x_k}}{alltraces}.
    \)
  
    \item \textbf{Rerouting update.}
    In this scenario, traffic is redirected from a device \(A_i\) to another device \(x_i\), where \(x_i\) is to be synthesized.
    Let \(R_i\) denote one such rerouting operation, where \(A_i\) is chosen randomly:
    \(
    R_i = \rmap{loc=A_i \cdot loc\leftarrow x_i}{alltraces}.
    \)
    For \(k\) parameters, we then define
    \(
    N'_{pre} = N_{pre}\triangleright R_1 \triangleright \cdots \triangleright R_k,
    \)
    representing a sequence of \(k\) rerouting operations.
\end{itemize}

For each choice of \(k\), each algorithm, and each update scenario, we sample \(500\) instances from the \(21112\) traffic flows.
The results are summarized in Table~\ref{tab:param-bench-full}.
As the table shows, most average running times are below \(1\) second, with the largest average among the successful runs being around \(10\) seconds for Delete Equiv.
Although some settings exhibit occasional outliers, such as Reroute Equiv.\ and Delete Weight, the only case that times out (>\(1000\)s) is Reroute Equiv.\ with \(k=3\).
This can be explained by the fact that, in a traffic flow containing \(30\)–\(40\) devices, a single rerouting may affect \(60\)–\(80\) links.
Applying such reroutings three times can therefore lead to a state explosion on the order of \(60^3\).
Aside from this case, the vast majority of queries finish within \(10\) seconds.

\paragraph{Batfish Tests.}
\begin{figure}[t]\vspace{-5pt}
\centering

\begin{subtable}[t]{0.45\linewidth}
\centering
\setlength{\tabcolsep}{4pt}
\begin{tabular}{|l|c|c|}
\hline
Query & None (s) & Dst.ip (s) \\
\hline
Traceroute 1 & 1.500 & 2.282 \\
Reachability 1 & 3.297 & 2.125 \\
Differential 1 & 2.938 & 7.343 \\
Reachability 2 & 2.046 & 2.734 \\
Differential 2 & 1.047 & 2.422 \\
Traceroute 2 & 0.282 & 1.016 \\
Traceroute 3 & 1.172 & 1.703 \\
Reachability 3 & 2.281 & 2.719 \\
Differential 3 & 2.360 & 4.859 \\
Reachability 4 & 2.266 & 2.672 \\
Differential 4 & 2.203 & 2.500 \\
\hline
Average & 1.945 & 2.943 \\
\hline
\end{tabular}
\caption{Change validation. (8.5k+ json)}\vspace{-5pt}
\label{tab:change-validation}
\end{subtable}
\hfill
\begin{subtable}[t]{0.53\linewidth}
\centering
\setlength{\tabcolsep}{4pt}
\begin{tabular}{|l|c|c|c|}
\hline
Query & None (s) & Loc (s) & Loc+Dst.ip (s) \\
\hline
Traceroute 1 & 51.687 & 69.093 & 114.781 \\
Traceroute 2 & 6.782  & 15.954 & 32.312 \\
Traceroute 3 & 5.922  & 23.546 & 45.204 \\
Traceroute 4 & 5.015  & 18.610 & 38.218 \\
Traceroute 5 & 6.500  & 33.234 & 66.219 \\
Reachability 1 & 19.375 & 19.031 & 34.219 \\
\hline
Average & 15.880 & 29.911 & 55.159 \\
\hline
\end{tabular}
\caption{Hybrid validation. (200k+ json)}\vspace{-5pt}
\label{tab:hybrid-validation}
\end{subtable}

\caption{Validation results with different middle packet inspection on Batfish scenarios.}\vspace{-10pt}
\label{fig:validation-results}
\end{figure}
Batfish \cite{batfish} is a state-of-the-art network analysis tool, and its tutorials provide a number of realistic examples together with their routing configurations.
In particular, we focus on the \href{https://batfish.readthedocs.io/en/latest/notebooks/linked/introduction-to-forwarding-change-validation.html}{\emph{Forwarding Change Validation}} and \href{https://pybatfish.readthedocs.io/en/latest/notebooks/linked/analyzing-public-and-hybrid-cloud-networks.html}{\emph{Hybrid Cloud Network}} benchmarks, which contain approximately \(8.5\)K and \(200\)K lines of routing information, respectively.
In these two tutorials Batfish users are asked to issue \texttt{reachability()} queries, which ask whether certain packets can reach certain locations under given path constraints, and \texttt{traceroute()} queries, which return the results of sending packets down network paths.  These tutorials also use
\texttt{differentialReachability()}, which asks how the set of locations reachable by a packet differs between two networks.
In our experiments, we use Parametric NetKAT to 
ask closely related queries that extract information necessary to solve the tutorials' challenges.

\OMIT{
In our experiments, we aim to reproduce these Batfish queries as closely as possible using Parametric NetKAT.
The main challenge is that Batfish does not merely return yes-or-no answers: it also provides sample traces witnessing the results of the queries.
By introducing parameters, we can reproduce the execution results of these queries as well as inspect packets at intermediate points along the path.}

Let \(N\) be the NetKAT expression obtained by parsing the Batfish routing tables.
The first step is to introduce the parameters of interest.
In the \href{https://batfish.readthedocs.io/en/latest/notebooks/linked/introduction-to-forwarding-change-validation.html}{\emph{Forwarding Change Validation}} benchmark, the goal is to install the correct filters, and the tutorial uses traces to inspect packets at a specific device.
Suppose this device is \(A\).
We reproduce this style of query using
\(
N \triangleright \rid{alltraces}\rconcat\rfilter{\lift{loc=A\pand dst.ip=x}}\rconcat\rid{alltraces}.
\)
In the \href{https://pybatfish.readthedocs.io/en/latest/notebooks/linked/analyzing-public-and-hybrid-cloud-networks.html}{\emph{Hybrid Cloud Network}} benchmark, the user is interested in what happens along the path.
For this purpose, we use
\(
N \triangleright \rid{alltraces}\rconcat\rfilter{\lift{loc=x}}\rconcat\rid{alltraces}
\)
and
\(
N \triangleright \rid{alltraces}\rconcat\rfilter{\lift{loc=x \pand dst.ip=y}}\rconcat\rid{alltraces}
\)
to represent queries over the topology path alone, and over the topology path together with packet IP information along the way.

After introducing the parameters, we obtain a parameterized network \(N'\).
With some additional post-processing to impose path or packet constraints on \(N'\), the remaining queries are easy to express.
The queries \texttt{reachability()} and \texttt{traceroute()} can be reproduced using emptiness queries, while \texttt{differentialReachability()} can be reproduced using equivalence queries.

The results are shown in Figure~\ref{fig:validation-results}.
The different columns indicate the packet-field information inspected during the query: \emph{None} denotes no inspection, \emph{Loc} denotes inspection of the location field, and \emph{Dst.ip} denotes inspection of the destination IP field.
As the figure shows, even on the benchmark with \(200\)K lines of routing information, querying additional along-the-way information increases running time by at most about a factor of two, while the average running time remains below one minute.
These results suggest that, even on industrial benchmarks with very large routing tables, parameterizing the queries introduces only modest overhead compared with the original queries, typically by at most a factor of two for packet inspection.
\section{Related Work} \label{sec:relatedwork}

Parametric NetKAT builds on a decade of research on NetKAT and related languages \cite{netkat,prob-netkat,Moeller24Katch,han26relational,acevedo26weight}.  In doing so, it inherits many of NetKAT's most useful properties:  Flexible network modeling capabilities, a compositional language design, a clear denotational semantics,
and efficient automata-theoretic decision procedures.  Parametric NetKAT advances the state of the art by combining multiple independent sublanguages (NetKAT, Weighted NetKAT, which subsumes Probabilistic NetKAT, and Relational NetKAT), extending the syntax of these sublanguages with parameters, and supplying algorithms that find valuations of those parameters under boolean combinations of constraints.  Whereas past NetKAT systems answered verification questions, Parametric NetKAT answers enumeration questions---it is a new kind of ``AllSAT solver,'' specialized for the networking domain. 

NetKAT is far from the only network verification framework available.  
Other network (dataplane) verification systems include AntEater~\cite{anteater}, Header Space Analysis (HSA)~\cite{hsa}, Veriflow~\cite{veriflow}, 
Atomic Predicates~\cite{atomic}, and
DeltaNet~\cite{deltanet}, among others.  
Whereas NetKAT emphasizes
compositional language design and clear semantics (though not to the exclusion of performance considerations~\cite{Moeller24Katch}), these other efforts 
explored specific implementation strategies and optimizations including the use of SMT~\cite{anteater}, smart data structures and representations~\cite{hsa,veriflow,atomic}, and incremental analysis~\cite{netplumber,deltanet}.

In terms of its high-level objective, Network Optimized Datalog (NOD)~\cite{nod} is perhaps the most closely related system to Parametric NetKAT.  NOD uses datalog to specify properties of network data planes, with datalog variables being used to record information about traces.  The NOD implementation uses Z3's datalog engine, retrofitted with new data structures to accelerate analysis in the
networking domain.  NOD can answer "which packets?" and "multipath differencing" queries.  
However, it was not clear how one would implement general-purpose equivalence checking effectively in NOD, as in our multi-objective synthesis example, nor whether one could accommodate quantitative objectives easily---such examples were not explored in the work on NOD.  Despite similarities in objective, from a technical standpoint, the systems are quite different:  Parametric NetKAT is based on Kleene Algebra with Tests~\cite{kat}, extended with regular relations, weights, and domain-specific operations. Parametric NetKAT has a denotational semantics based on sets of traces, while NOD is a variant of datalog and does not come with a domain-specific formal semantics of its own.  Parametric NetKAT is compiled into nonstandard automata, deploys domain-specific optimizations, and uses variants of emptiness testing and weighted trace computation algorithms to compute valid parameter spaces. 

Batfish~\cite{batfish} is another powerful network analysis system that straddles control plane and data plane analysis. It simulates the network control plane in a customizable environment, producing a data plane, and then allows a variety of data plane analyses to examine the results.
Batfish is particularly flexible as it provides a very rich set of built-in queries for network analysis, and allows users to build more if they are willing to program with the internals in Java. 

Less closely related research includes research on pure network control-plane simulation, verification, debugging, and synthesis, such as work on   MineSweeper~\cite{minesweeper},
Campion~\cite{campion}, NV~\cite{nv}, NetComplete~\cite{netcomplete}, and Expresso~\cite{expresso} among others.  Elements of the network control plane exchange messages to decide which routes to use, whereas the network data plane implements those routes.  These semantic differences often lead to different kinds of models and different algorithms for analyzing those models.  Having said that, both NV and Expresso engage in symbolic analysis of network control planes by extending the route announcement space with symbolic representations of possible valuations, so there are some commonalities in implementation techniques between such systems and Parametric NetKAT, though the specification mechanisms are quite different. 

Our algorithms for parametric model checking draw upon techniques for symbolic model checking \cite{SMC}, particularly the use of binary decision diagrams to represent state spaces symbolically.  Parametric NetKAT is similar in spirit to parametric model checking \cite{PMC}, in which the input is a discrete-time Markov chain parameterized by transition probabilities, and the goal is to synthesize a description of the set of parameter valuations under which the Markov chain satisfies a given property of interest.  However, the techniques are quite different---in parametric NetKAT parameters are discrete, and we use binary decision diagrams to represent sets of parameter valuations.

\section{Conclusions}

Parametric NetKAT is a new domain-specific language for asking
\emph{enumeration questions} about networks.  It combines elements of
NetKAT, Relational NetKAT, and Weighted NetKAT and extends them with parameters.  The Parametric NetKAT solver generates valuations for parameters that satisfy constraints involving set emptiness, equivalence, and/or weights by compiling parametric expressions into non-parametric automata that operate over an extended packet space, and using modified emptiness and weight aggregation algorithms to analyze them.  We demonstrate the utility of Parametric NetKAT by crafting a range of interesting network diagnosis queries including those asking "which packets" travel along a set of paths, "which failures" might be the root cause of observed outages, and "which processing differences" exist along two paths that should process packets similarly, among others.  We evaluate the performance of the system on a range of benchmarks drawn from industrial sources, including Batfish configurations with up to \(200\)K lines and Topology Zoo networks with up to \(700\) devices.
Our results show that the system can answer useful queries on these benchmarks efficiently, typically within tens of seconds, and at worst within roughly \(100\)–\(200\) seconds.

\bibliography{reference,more-references}


\newpage
\appendix
\section{Appendix}
\subsection{Defintions of Section 4}
\paragraph{NetKAT Transducer}

NetKAT transducers model the semantics of relational NetKAT programs \( R \), which follows directly from the prior work \cite{han26relational}.
\begin{definition}
    A \textbf{NetKAT transducer} is a tuple $T = (S,S_0,S_f,\Delta_S,\Delta_L,\Delta_R,\Delta_E)$ where $S$ is a finite set of states, $S_0 \subseteq S$ is a set of start states,
    $S_f \subseteq S$ is a set of final states,
    $\Delta_S : S \times S \rightarrow (\textit{Pk} \times \textit{Pk}) \times (\textit{Pk} \times \textit{Pk})$ is a \textit{synchronous} transition relation,
    $\Delta_L : S \times S \rightarrow \textit{Pk} \times \textit{Pk}$ is an \textit{asynchronous left} transition relation, 
    $\Delta_R : S \times S \rightarrow \textit{Pk} \times \textit{Pk}$ is an \textit{asynchronous right} transition relation, and
    $\Delta_E : S \times S \rightarrow \textit{Pk} \times \textit{Pk}$ is an \textit{epsilon} transition relation.
\end{definition}

Similar to NetKAT automaton $M$, the transition relation of $T$ is defined inductively over pairs of packet traces. A labeled transition of the form:
\[
(s_0, (x_0, y_0)) \stackrel{(\tau_1, \tau_2)}{\longrightarrow} (s_n, (x_{n_1}, y_{n_2}))
\]
indicates that, starting from $(x_0, y_0)$, the transducer produces the output traces $\tau_1$ and $\tau_2$ along a path to state $s_n$ with final packets $(x_{n_1}, y_{n_2})$.

\begin{itemize}
    \item \textbf{Base case:}
    $
    (s_0, (x_0, y_0)) \stackrel{(\epsilon, \epsilon)}{\longrightarrow} (s_0, (x_0, y_0)).
    $
    
    \item \textbf{Both tapes move:} If $((x_0, y_0), (x_1, y_1)) \in \Delta_S(s_0, s_1)$ and 
   $ (s_1, (x_1, y_1)) \stackrel{(\tau_1, \tau_2)}{\longrightarrow} (s_n, (x_{n_1}, y_{n_2}))$,
    then
    $
    (s_0, (x_0, y_0))\xrightarrow[]{(x_1 \tau_1, y_1 \tau_2)}  (s_n, (x_{n_1}, y_{n_2})).
    $
    
    \item \textbf{First tape only:} If $(x_0, x_1) \in \Delta_L(s_0, s_1)$ and 
    $
    (s_1, (x_1, y_0)) \stackrel{(\tau_1, \tau_2)}{\longrightarrow} (s_n, (x_{n_1}, y_{n_2})),
    $
    
    then
    $
    (s_0, (x_0, y_0))\xrightarrow[]{(x_1 \tau_1, \tau_2)}  (s_n, (x_{n_1}, y_{n_2})).
    $
    
    \item \textbf{Second tape only:} If $(y_0, y_1) \in \Delta_R(s_0, s_1)$ and 
    $
    (s_1, (x_0, y_1)) \stackrel{(\tau_1, \tau_2)}{\longrightarrow} (s_n, (x_{n_1}, y_{n_2})),
    $
    
    then
    $
    (s_0, (x_0, y_0))\xrightarrow[]{(\tau_1, y_1 \tau_2)}  (s_n, (x_{n_1}, y_{n_2})).
    $

    \item \textbf{No tape moves:} If $(x_0, y_0) \in \Delta_E(s_0, s_1)$ and 
   $ (s_1, (x_0, y_0)) \stackrel{(\tau_1, \tau_2)}{\longrightarrow} (s_n, (x_{n_1}, y_{n_2}))$,

    then
    $
    (s_0, (x_0, y_0)) \stackrel{(\tau_1, \tau_2)}{\longrightarrow} (s_n, (x_{n_1}, y_{n_2})).
    $
\end{itemize}

The language accepted by the NetKAT transducer $T$ is then defined as the set of trace pairs processed from a start state to an accepting state:
\[
L(T) = \{(x\tau_1, y\tau_2) \mid  s_0 \in S_0, s_f \in S_f, x,x',y,y'\in \mathit{Pk}.(s_0, (x, y)) \stackrel{(\tau_1, \tau_2)}{\longrightarrow} (s_f, (x', y'))\}.
\]

\section{Proof of Section 5}
\subsection{Full Homomorphism Theorem}
\begin{theorem}[Homomorphism]\label{thm:full-homo}
The alternative interpretation is homomorphic with respect to the semantic constructors of Parametric NetKAT, Relational NetKAT, and Weighted NetKAT.

\paragraph{Parametric NetKAT}
For all valuations \(v\in V\), the following hold:
\begin{align*}
\ipn{\dpn{PkR}}
  &= \{\, (pk_1,v)(pk_2,v) \mid ((pk_1,v),(pk_2,v)) \in \ipkr{\dpkr{PkR}} \,\}, \\
\ipn{\dpn{PN \triangleright RN}}
  &= \{\, \tau_2 \mid \tau_1 \in \ipn{\dpn{PN}},\ (\tau_1,\tau_2)\in \irn{\dpr{RN}},\ |\tau_2|\ge 2 \,\}, \\
\ipn{\dpn{PN_1 + PN_2}}
  &= \ipn{\dpn{PN_1}} \cup \ipn{\dpn{PN_2}}, \\
\ipn{\dpn{PN_1 \kcap PN_2}}
  &= \ipn{\dpn{PN_1}} \cap \ipn{\dpn{PN_2}}, \\
\ipn{\dpn{PN_1 \kdiff PN_2}}
  &= \ipn{\dpn{PN_1}} \setminus \ipn{\dpn{PN_2}}, \\
\ipn{\dpn{PN_1 \kand PN_2}}
  &= \{\, \tau \mid \exists \tau_1 \in \ipn{\dpn{PN_1}},\ \exists \tau_2 \in \ipn{\dpn{PN_2}},\ \tau = \tau_1 \kand \tau_2 \,\}, \\
\ipn{\dpn{\ks{PN}}}
  &= \bigcup_{n\ge 0} \ipn{\dpn{PN^n}},
  \qquad \text{where } PN^0 = \idpkr,\quad PN^{n+1}=PN^n \kand PN, \\
\ipn{\dpn{dup}}
  &= \{\, (pk,v)(pk,v)(pk,v) \mid pk\in \mathit{Pk},\ v\in V \,\}.
\end{align*}

\paragraph{Relational NetKAT}
\begin{align*}
\irn{\dpr{\rfilter{PkR}}}
  &= \{\, (\tau_1,\tau_2) \mid (\tau_1,\tau_2)\in \ipkr{\dpkr{PkR}} \,\}, \\
\irn{\dpr{\rmap{PkR}{PN}}}
  &= \{\, \bigl((pk_1,v)\trcat\cdots\trcat(pk_n,v),\;
         (pk_1',v)\trcat\cdots\trcat(pk_n',v)\bigr) \mid \\
  &\qquad (pk_1,v)\trcat\cdots\trcat(pk_n,v)\in \ipn{\dpn{PN}}, \\
  &\qquad \forall i\in[1,n].\ ((pk_i,v),(pk_i',v))\in \ipkr{\dpkr{PkR}}
     \,\}, \\
\irn{\dpr{\rdelete{PN}}}
  &= \{\, \bigl((pk_1,v)\trcat\cdots\trcat(pk_n,v),\; (pk,v)\bigr) \mid \\
  &\qquad pk\in Pk,\;
     (pk_1,v)\trcat\cdots\trcat(pk_n,v)\in \ipn{\dpn{PN}}
     \,\}, \\
\irn{\dpr{\rinsert{PN}}}
  &= \{\, \bigl((pk,v),\; (pk_1,v)\trcat\cdots\trcat(pk_n,v)\bigr) \mid \\
  &\qquad pk\in Pk,\;
     (pk_1,v)\trcat\cdots\trcat(pk_n,v)\in \ipn{\dpn{PN}}
     \,\}, \\
\irn{\dpr{RN_1 + RN_2}}
  &= \irn{\dpr{RN_1}} \cup \irn{\dpr{RN_2}}, \\
\irn{\dpr{RN_1 \rconcat RN_2}}
  &= \{\, (\tau_1 \trcat \tau_2,\; \tau_1' \trcat \tau_2') \mid \\
  &\qquad (\tau_1,\tau_1') \in \irn{\dpr{RN_1}}, \\
  &\qquad (\tau_2,\tau_2') \in \irn{\dpr{RN_2}}
     \,\}, \\
\irn{\dpr{RN^*}}
  &= \bigcup_{n\ge 0} \irn{\dpr{RN^n}}, \\
  &\qquad \text{where } RN^0 = \rfilter{\havoc},
     \quad RN^{n+1}=RN^n \rconcat RN.
\end{align*}
\paragraph{Weighted NetKAT}
\begin{align*}
\iwn{\dw{PN}}
  &= \{\, ((pk_1,v)\cdots(pk_n,v),1) \mid
        (pk_1,v)\cdots(pk_n,v)\in \ipn{\dpn{PN}} \,\} \\
  &\qquad \cup\; \{\, (\tau,0) \mid \tau\notin \ipn{\dpn{PN}} \,\}, \\
\iwn{\dw{PN\triangleright WN}}
  &= \{\, (\tau,w) \mid
        \tau\in \ipn{\dpn{PN}} \land (\tau,w)\in \iwn{\dw{WN}} \,\} \\
  &\qquad \cup\; \{\, (\tau,0) \mid \tau\notin \ipn{\dpn{PN}} \,\}, \\
\iwn{\dw{w\wtimes WN}}
  &= \{\, (\tau,w\cdot w') \mid (\tau,w')\in \iwn{\dw{WN}} \,\}, \\
\iwn{\dw{WN\wtimes w}}
  &= \{\, (\tau,w'\cdot w) \mid (\tau,w')\in \iwn{\dw{WN}} \,\}, \\
\iwn{\dw{WN_1\wor WN_2}}
  &= \{\, (\tau,w_1+w_2) \mid
        (\tau,w_1)\in \iwn{\dw{WN_1}},\;
        (\tau,w_2)\in \iwn{\dw{WN_2}} \,\}, \\
\iwn{\dw{WN_1\wand WN_2}}
  &= \left\{\, \left(\tau,
      \sum_{\substack{
      (\tau_1,w_1)\in \iwn{\dw{WN_1}} \\
      (\tau_2,w_2)\in \iwn{\dw{WN_2}} \\
      \tau=\tau_1\kand\tau_2
      }} w_1 w_2\right)
      \,\right\}, \\
\iwn{\dw{WN^*}}
  &= \left\{\, \left(\tau,
      \sum_{(\tau,w)\in \iwn{\dw{WN^n}}} w\right)
      \,\middle|\, n\ge 0 \right\}, \\
  &\qquad \text{where } WN^0=\lift{1},
     \quad WN^{n+1}=WN^n\wand WN.
\end{align*}\end{theorem}

\begin{proof}
All equalities follow by unfolding the definitions.

\begin{align*}
\ipn{\dpn{PkR}}
&= \{\, (pk_1,v)(pk_2,v) \mid pk_1pk_2 \in \dpn{PkR}(v) \,\} \\
&= \{\, (pk_1,v)(pk_2,v) \mid (pk_1,pk_2)\in \dpkr{PkR}(v) \,\} \\
&= \{\, (pk_1,v)(pk_2,v) \mid ((pk_1,v),(pk_2,v)) \in \ipkr{\dpkr{PkR}} \,\}.
\end{align*}

\begin{align*}
\ipn{\dpn{PN \triangleright RN}}
&= \{\, \tau_2 \mid \tau_2 = (pk_1',v)\cdots(pk_m',v),\ pk_1'\cdots pk_m' \in \dpn{PN\triangleright RN}(v) \,\} \\
&= \{\, \tau_2 \mid \tau_2 = (pk_1',v)\cdots(pk_m',v),\ \tau_1' \in \dpn{PN}(v), \\
&\hspace{7em} (\tau_1',pk_1'\cdots pk_m')\in \dpr{RN}(v),\ m\ge 2 \,\} \\
&= \{\, \tau_2 \mid \tau_1 \in \ipn{\dpn{PN}},\ (\tau_1,\tau_2)\in \irn{\dpr{RN}},\ |\tau_2|\ge 2 \,\}.
\end{align*}

\begin{align*}
\ipn{\dpn{PN_1 + PN_2}}
&= \{\, \tau \mid \tau=(pk_1,v)\cdots(pk_n,v),\ pk_1\cdots pk_n \in \dpn{PN_1+PN_2}(v) \,\} \\
&= \{\, \tau \mid \tau=(pk_1,v)\cdots(pk_n,v),\ pk_1\cdots pk_n \in \dpn{PN_1}(v)\cup \dpn{PN_2}(v) \,\} \\
&= \ipn{\dpn{PN_1}} \cup \ipn{\dpn{PN_2}}.
\end{align*}

\begin{align*}
\ipn{\dpn{PN_1 \kcap PN_2}}
&= \{\, \tau \mid \tau=(pk_1,v)\cdots(pk_n,v),\ pk_1\cdots pk_n \in \dpn{PN_1\kcap PN_2}(v) \,\} \\
&= \{\, \tau \mid \tau=(pk_1,v)\cdots(pk_n,v),\ pk_1\cdots pk_n \in \dpn{PN_1}(v)\cap \dpn{PN_2}(v) \,\} \\
&= \ipn{\dpn{PN_1}} \cap \ipn{\dpn{PN_2}}.
\end{align*}

\begin{align*}
\ipn{\dpn{PN_1 \kdiff PN_2}}
&= \{\, \tau \mid \tau=(pk_1,v)\cdots(pk_n,v),\ pk_1\cdots pk_n \in \dpn{PN_1\kdiff PN_2}(v) \,\} \\
&= \{\, \tau \mid \tau=(pk_1,v)\cdots(pk_n,v),\ pk_1\cdots pk_n \in \dpn{PN_1}(v)\setminus \dpn{PN_2}(v) \,\} \\
&= \ipn{\dpn{PN_1}} \setminus \ipn{\dpn{PN_2}}.
\end{align*}

\begin{align*}
\ipn{\dpn{PN_1 \kand PN_2}}
&= \{\, \tau \mid \tau=(pk_1,v)\cdots(pk_n,v),\ pk_1\cdots pk_n \in \dpn{PN_1\kand PN_2}(v) \,\} \\
&= \{\, \tau \mid \exists \tau_1' \in \dpn{PN_1}(v),\ \exists \tau_2' \in \dpn{PN_2}(v),\ pk_1\cdots pk_n = \tau_1' \kand \tau_2' \,\} \\
&= \{\, \tau \mid \exists \tau_1 \in \ipn{\dpn{PN_1}},\ \exists \tau_2 \in \ipn{\dpn{PN_2}},\ \tau = \tau_1 \kand \tau_2 \,\}.
\end{align*}

\begin{align*}
\ipn{\dpn{\ks{PN}}}
&= \{\, \tau \mid \tau=(pk_1,v)\cdots(pk_n,v),\ pk_1\cdots pk_n \in \dpn{\ks{PN}}(v) \,\} \\
&= \{\, \tau \mid \tau=(pk_1,v)\cdots(pk_n,v),\ pk_1\cdots pk_n \in \bigcup_{n\ge 0}\dpn{PN^n}(v) \,\} \\
&= \bigcup_{n\ge 0} \ipn{\dpn{PN^n}},
\qquad \text{where } PN^0=\idpkr,\ PN^{n+1}=PN^n\kand PN.
\end{align*}

\begin{align*}
\ipn{\dpn{dup}}
&= \{\, (pk_1,v)(pk_2,v)\cdots(pk_n,v) \mid pk_1pk_2\cdots pk_n \in \dpn{dup}(v) \,\} \\
&= \{\, (pk,v)(pk,v)(pk,v) \mid pk\in \mathit{Pk},\ v\in V \,\}.
\end{align*}

\begin{align*}
\irn{\dpr{\rfilter{PkR}}}
&=
\left\{
\bigl((pk_1,v),\;(pk_2,v)\bigr)
\;\middle|\;
(pk_1,pk_2)\in \dpr{\rfilter{PkR}}(v)
\right\} \\
&=
\left\{
\bigl((pk_1,v),\;(pk_2,v)\bigr)
\;\middle|\;
(pk_1,pk_2)\in \dpkr{PkR}(v)
\right\} \\
&=
\{\, (\tau_1,\tau_2) \mid (\tau_1,\tau_2)\in \ipkr{\dpkr{PkR}} \,\}.
\end{align*}

\begin{align*}
\irn{\dpr{\rmap{PkR}{PN}}}
&=
\left\{
\bigl((pk_1,v)\trcat\cdots\trcat(pk_n,v),\;
      (pk_1',v)\trcat\cdots\trcat(pk_n',v)\bigr)
\;\middle|\;
\right.\\
&\qquad\left.
(pk_1\trcat\cdots\trcat pk_n,\;
 pk_1'\trcat\cdots\trcat pk_n')
\in \dpr{\rmap{PkR}{PN}}(v)
\right\} \\
&=
\left\{
\bigl((pk_1,v)\trcat\cdots\trcat(pk_n,v),\;
      (pk_1',v)\trcat\cdots\trcat(pk_n',v)\bigr)
\;\middle|\;
\right.\\
&\qquad\left.
pk_1\trcat\cdots\trcat pk_n \in \dpn{PN}(v),\;
\forall i\in[1,n].\ (pk_i,pk_i')\in \dpkr{PkR}(v)
\right\} \\
&=
\left\{
\bigl((pk_1,v)\trcat\cdots\trcat(pk_n,v),\;
      (pk_1',v)\trcat\cdots\trcat(pk_n',v)\bigr)
\;\middle|\;
\right.\\
&\qquad\left.
(pk_1,v)\trcat\cdots\trcat(pk_n,v)\in \ipn{\dpn{PN}},
\right.\\
&\qquad\left.
\forall i\in[1,n].\ ((pk_i,v),(pk_i',v))\in \ipkr{\dpkr{PkR}}
\right\}.
\end{align*}

\begin{align*}
\irn{\dpr{\rdelete{PN}}}
&=
\left\{
\bigl((pk_1,v)\trcat\cdots\trcat(pk_n,v),\;(pk,v)\bigr)
\;\middle|\;
(pk_1\trcat\cdots\trcat pk_n,\;pk)\in \dpr{\rdelete{PN}}(v)
\right\} \\
&=
\left\{
\bigl((pk_1,v)\trcat\cdots\trcat(pk_n,v),\;(pk,v)\bigr)
\;\middle|\;
pk\in Pk,\;
pk_1\trcat\cdots\trcat pk_n \in \dpn{PN}(v)
\right\} \\
&=
\left\{
\bigl((pk_1,v)\trcat\cdots\trcat(pk_n,v),\;(pk,v)\bigr)
\;\middle|\;
pk\in Pk,\;
(pk_1,v)\trcat\cdots\trcat(pk_n,v)\in \ipn{\dpn{PN}}
\right\}.
\end{align*}

\begin{align*}
\irn{\dpr{\rinsert{PN}}}
&=
\left\{
\bigl((pk,v),\;(pk_1,v)\trcat\cdots\trcat(pk_n,v)\bigr)
\;\middle|\;
(pk,\;pk_1\trcat\cdots\trcat pk_n)\in \dpr{\rinsert{PN}}(v)
\right\} \\
&=
\left\{
\bigl((pk,v),\;(pk_1,v)\trcat\cdots\trcat(pk_n,v)\bigr)
\;\middle|\;
pk\in Pk,\;
pk_1\trcat\cdots\trcat pk_n \in \dpn{PN}(v)
\right\} \\
&=
\left\{
\bigl((pk,v),\;(pk_1,v)\trcat\cdots\trcat(pk_n,v)\bigr)
\;\middle|\;
pk\in Pk,\;
(pk_1,v)\trcat\cdots\trcat(pk_n,v)\in \ipn{\dpn{PN}}
\right\}.
\end{align*}

\begin{align*}
\irn{\dpr{RN_1+RN_2}}
&=
\left\{
(\tau_1,\tau_2)
\;\middle|\;
\exists v,\exists \sigma_1,\exists \sigma_2.\right.\\
&\qquad
(\sigma_1,\sigma_2)\in \dpr{RN_1+RN_2}(v),\\
&\qquad
\tau_1=(pk_{11},v)\cdots(pk_{1n_1},v),\;
\tau_2=(pk_{21},v)\cdots(pk_{2n_2},v),\\
&\qquad\left.
(\sigma_1,\sigma_2)=
(pk_{11}\cdots pk_{1n_1},\;
 pk_{21}\cdots pk_{2n_2})
\right\} \\
&=
\left\{
(\tau_1,\tau_2)
\;\middle|\;
\exists v,\exists \sigma_1,\exists \sigma_2.\right.\\
&\qquad
(\sigma_1,\sigma_2)\in \dpr{RN_1}(v)\cup\dpr{RN_2}(v),\\
&\qquad
\tau_1=(pk_{11},v)\cdots(pk_{1n_1},v),\;
\tau_2=(pk_{21},v)\cdots(pk_{2n_2},v),\\
&\qquad\left.
(\sigma_1,\sigma_2)=
(pk_{11}\cdots pk_{1n_1},\;
 pk_{21}\cdots pk_{2n_2})
\right\} \\
&=
\irn{\dpr{RN_1}} \cup \irn{\dpr{RN_2}}.
\end{align*}

\begin{align*}
\irn{\dpr{RN_1\rconcat RN_2}}
&=
\left\{
(\tau,\tau')
\;\middle|\;
\exists v,\exists \sigma,\exists \sigma'.\right.\\
&\qquad
(\sigma,\sigma')\in \dpr{RN_1\rconcat RN_2}(v),\\
&\qquad
\tau=(pk_1,v)\trcat\cdots\trcat(pk_{n_2},v),\\
&\qquad
\tau'=(pk_1',v)\trcat\cdots\trcat(pk_{m_2}',v),\\
&\qquad\left.
(\sigma,\sigma')=
(pk_1\trcat\cdots\trcat pk_{n_2},\;
 pk_1'\trcat\cdots\trcat pk_{m_2}')
\right\} \\
&=
\left\{
(\tau_1\trcat\tau_2,\;\tau_1'\trcat\tau_2')
\;\middle|\;
(\tau_1,\tau_1')\in \irn{\dpr{RN_1}},\right.\\
&\qquad\left.
(\tau_2,\tau_2')\in \irn{\dpr{RN_2}}
\right\}.
\end{align*}

\begin{align*}
\irn{\dpr{RN^*}}
&=
\left\{
(\tau,\tau')
\;\middle|\;
\exists v,\exists \sigma,\exists \sigma'.\right.\\
&\qquad
(\sigma,\sigma')\in \dpr{RN^*}(v),\\
&\qquad
\tau=(pk_{11},v)\cdots(pk_{1n_1},v),\\
&\qquad
\tau'=(pk_{21},v)\cdots(pk_{2n_2},v),\\
&\qquad\left.
(\sigma,\sigma')=
(pk_{11}\cdots pk_{1n_1},\;
 pk_{21}\cdots pk_{2n_2})
\right\} \\
&=
\left\{
(\tau,\tau')
\;\middle|\;
\exists v,\exists \sigma,\exists \sigma'.\right.\\
&\qquad
(\sigma,\sigma')\in \bigcup_{n\ge 0}\dpr{RN^n}(v),\\
&\qquad
\tau=(pk_{11},v)\cdots(pk_{1n_1},v),\\
&\qquad
\tau'=(pk_{21},v)\cdots(pk_{2n_2},v),\\
&\qquad\left.
(\sigma,\sigma')=
(pk_{11}\cdots pk_{1n_1},\;
 pk_{21}\cdots pk_{2n_2})
\right\} \\
&=
\bigcup_{n\ge 0}\irn{\dpr{RN^n}},
\qquad
\text{where } RN^0=\rfilter{\havoc},
\quad RN^{n+1}=RN^n\rconcat RN.
\end{align*}

\begin{align*}
\iwn{\dw{PN}}
&=
\{\, ((pk_1,v)\cdots(pk_n,v),w)
\mid (pk_1\cdots pk_n,w)\in \dw{PN}(v)\,\} \\
&=
\{\, ((pk_1,v)\cdots(pk_n,v),1)
\mid pk_1\cdots pk_n\in \dpn{PN}(v)\,\} \\
&\qquad \cup\;
\{\, (\tau,0)
\mid \tau\notin \ipn{\dpn{PN}} \,\} \\
&=
\{\, ((pk_1,v)\cdots(pk_n,v),1) \mid
      (pk_1,v)\cdots(pk_n,v)\in \ipn{\dpn{PN}} \,\} \\
&\qquad \cup\;
\{\, (\tau,0) \mid \tau\notin \ipn{\dpn{PN}} \,\}.
\end{align*}

\begin{align*}
\iwn{\dw{PN\triangleright WN}}
&=
\iwn{\dw{\wres{WN}{PN}}} \\
&=
\{\, ((pk_1,v)\cdots(pk_n,v),w)
\mid (pk_1\cdots pk_n,w)\in \dw{\wres{WN}{PN}}(v)\,\} \\
&=
\{\, ((pk_1,v)\cdots(pk_n,v),w)
\mid pk_1\cdots pk_n\in \dpn{PN}(v), \\
&\qquad\qquad\qquad\qquad
(pk_1\cdots pk_n,w)\in \dw{WN}(v)\,\} \\
&\qquad \cup\;
\{\, (\tau,0)\mid \tau\notin \ipn{\dpn{PN}} \,\} \\
&=
\{\, (\tau,w) \mid
      \tau\in \ipn{\dpn{PN}} \land (\tau,w)\in \iwn{\dw{WN}} \,\} \\
&\qquad \cup\;
\{\, (\tau,0) \mid \tau\notin \ipn{\dpn{PN}} \,\}.
\end{align*}

\begin{align*}
\iwn{\dw{w\wtimes WN}}
&=
\{\, ((pk_1,v)\cdots(pk_n,v),w'')
\mid (pk_1\cdots pk_n,w'')\in \dw{w\wtimes WN}(v)\,\} \\
&=
\{\, ((pk_1,v)\cdots(pk_n,v),w\cdot w')
\mid (pk_1\cdots pk_n,w')\in \dw{WN}(v)\,\} \\
&=
\{\, (\tau,w\cdot w') \mid (\tau,w')\in \iwn{\dw{WN}} \,\}.
\end{align*}

\begin{align*}
\iwn{\dw{WN\wtimes w}}
&=
\{\, ((pk_1,v)\cdots(pk_n,v),w'')
\mid (pk_1\cdots pk_n,w'')\in \dw{WN\wtimes w}(v)\,\} \\
&=
\{\, ((pk_1,v)\cdots(pk_n,v),w'\cdot w)
\mid (pk_1\cdots pk_n,w')\in \dw{WN}(v)\,\} \\
&=
\{\, (\tau,w'\cdot w) \mid (\tau,w')\in \iwn{\dw{WN}} \,\}.
\end{align*}

\begin{align*}
\iwn{\dw{WN_1\wor WN_2}}
&=
\{\, ((pk_1,v)\cdots(pk_n,v),w)
\mid (pk_1\cdots pk_n,w)\in \dw{WN_1\wor WN_2}(v)\,\} \\
&=
\{\, ((pk_1,v)\cdots(pk_n,v),w_1+w_2)
\mid (pk_1\cdots pk_n,w_1)\in \dw{WN_1}(v), \\
&\qquad\qquad\qquad\qquad
(pk_1\cdots pk_n,w_2)\in \dw{WN_2}(v)\,\} \\
&=
\{\, (\tau,w_1+w_2) \mid
      (\tau,w_1)\in \iwn{\dw{WN_1}},\;
      (\tau,w_2)\in \iwn{\dw{WN_2}} \,\}.
\end{align*}

\begin{align*}
\iwn{\dw{WN_1\wand WN_2}}
&=
\left\{\,
\left((pk_1,v)\cdots(pk_n,v),
\sum_{\substack{
(\tau_1,w_1)\in \dw{WN_1}(v) \\
(\tau_2,w_2)\in \dw{WN_2}(v) \\
pk_1\cdots pk_n=\tau_1\kand\tau_2
}} w_1w_2\right)
\,\right\} \\
&=
\left\{\,
\left(\tau,
\sum_{\substack{
(\tau_1,w_1)\in \iwn{\dw{WN_1}} \\
(\tau_2,w_2)\in \iwn{\dw{WN_2}} \\
\tau=\tau_1\kand\tau_2
}} w_1w_2\right)
\,\right\}.
\end{align*}

\begin{align*}
\iwn{\dw{WN^*}}
&=
\left\{\,
((pk_1,v)\cdots(pk_n,v),w)
\;\middle|\;
(pk_1\cdots pk_n,w)\in \dw{WN^*}(v)
\,\right\} \\
&=
\left\{\,
\left((pk_1,v)\cdots(pk_n,v),
\sum_{(pk_1\cdots pk_n,w)\in \dw{WN^n}(v)} w\right)
\;\middle|\; n\ge 0
\,\right\} \\
&=
\left\{\,
\left(\tau,
\sum_{(\tau,w)\in \iwn{\dw{WN^n}}} w\right)
\;\middle|\; n\ge 0
\,\right\}, \\
&\qquad
\text{where } WN^0=\lift{1},
\quad WN^{n+1}=WN^n\wand WN.
\end{align*}

This completes the proof.
\end{proof}
\begin{theorem}[\(PkR\) Construction]
For every parametric NetKAT expression \(PkR\), we can construct a NetKAT automaton \(M\) such that
\[
L(M)=\ipn{\dpn{PkR}}.
\]
\end{theorem}

\begin{proof}
Let
\[
M=(\{s_0,s_1\},S_0,S_f,\Delta),
\]
where \(s_0\) is the initial state and \(s_1\) is the accepting state.
More precisely, let
\[
S_0=\{s_0\},
\qquad
S_f=\{s_1\},
\]
and define the transition relation by
\[
\Delta(s_0,s_1)=\ipkr{\dpkr{PkR}},
\]
while
\[
\Delta(s_0,s_0)=
\Delta(s_1,s_0)=
\Delta(s_1,s_1)=\emptyset.
\]
Then \(M\) accepts exactly the traces corresponding to \(\ipn{\dpn{PkR}}\), as required.
\end{proof}

\subsection{Emptiness Check}

\begin{theorem}
The emptiness-checking algorithm in Algorithm~\ref{alg:emptiness} is correct with respect to its output specification.
\end{theorem}

\begin{proof}
Let \(M=(S,S_0,S_f,\Delta)\) be the input automaton, and recall that
\[
L(M)=\ipn{\dpn{PN}}.
\]
We must show that the algorithm returns exactly the set
\[
\{\, v \mid \dpn{PN}(v)\neq \emptyset \,\}.
\]

We first establish the following invariant: after every iteration of the algorithm, for every state \(s\in S\),
\[
Reach(s)\subseteq \{\, (pk,v)\mid (s_0,(pk_0,v)) \to^* (s,(pk,v)) \text{ for some } s_0\in S_0,\ pk_0\in Pk \,\},
\]
that is, every pair in \(Reach(s)\) is genuinely reachable from some initial state.
This holds initially by definition of \(Reach\): if \(s\in S_0\), then \(Reach(s)=Pk\times V\), which corresponds exactly to the possible initial configurations; otherwise \(Reach(s)=\emptyset\).
Moreover, every update in Step~(2)(b) adds a pair \((pk_2,v)\) to \(Reach(s')\) only when there exists \((pk_1,v)\in Reach(s)\) and
\[
((pk_1,v),(pk_2,v))\in \Delta(s,s').
\]
Hence every newly added element is reachable by one further transition.
Thus the invariant is preserved.

At termination, since the algorithm iterates until no \(Reach(s)\) changes, \(Reach\) is the least fixed point generated by the transition relation \(\Delta\), restricted to valuations not already in \(V_{\mathit{sol}}\).
Because the translated semantics preserves valuations along transitions, every run of \(M\) has the form
\[
(s_0,(pk_0,v)) \to (s_1,(pk_1,v)) \to \cdots \to (s_n,(pk_n,v)),
\]
with the same valuation \(v\) throughout.
This is exactly the content of the translated semantics \(L(M)=\ipn{\dpn{PN}}\).

We now prove soundness and completeness.

\emph{Soundness.}
Suppose \(v\in V_{\mathit{sol}}\) when the algorithm terminates.
By construction, \(v\) is added to \(V_{\mathit{sol}}\) only in Step~(2)(a), so there must exist some final state \(s\in S_f\) and some packet \(pk\in Pk\) such that
\[
(pk,v)\in Reach(s).
\]
By the reachability invariant above, there is a run of \(M\) from some initial state to \(s\) ending in \((pk,v)\).
Since \(s\in S_f\), this run is accepting.
Therefore there exists a trace in \(L(M)\) whose valuation component is \(v\).
Using \(L(M)=\ipn{\dpn{PN}}\), this means exactly that
\[
\dpn{PN}(v)\neq \emptyset.
\]
Hence every valuation returned by the algorithm satisfies the output specification.

\emph{Completeness.}
Suppose now that
\[
\dpn{PN}(v)\neq \emptyset.
\]
Then, since \(L(M)=\ipn{\dpn{PN}}\), there exists an accepting run of \(M\) labeled by some trace
\[
(pk_0,v)(pk_1,v)\cdots(pk_n,v)
\]
ending in a final state \(s_n\in S_f\).
We show by induction on the length of this run that each configuration along the run is eventually added to the corresponding \(Reach\) set.

For the initial configuration, since the run starts in some \(s_0\in S_0\), we have
\[
(pk_0,v)\in Reach(s_0)
\]
from initialization.
Now assume \((pk_i,v)\in Reach(s_i)\) has been added.
Because the run follows a transition
\[
((pk_i,v),(pk_{i+1},v))\in \Delta(s_i,s_{i+1}),
\]
Step~(2)(b) will eventually add \((pk_{i+1},v)\) to \(Reach(s_{i+1})\), unless \(v\) has already been placed in \(V_{\mathit{sol}}\).
But in that case we are already done.
Thus, by induction, either \(v\) is already in \(V_{\mathit{sol}}\), or eventually \((pk_n,v)\in Reach(s_n)\) for the accepting state \(s_n\in S_f\).
Then Step~(2)(a) adds \(v\) to \(V_{\mathit{sol}}\).

Therefore every valuation \(v\) such that \(\dpn{PN}(v)\neq\emptyset\) is eventually returned by the algorithm.

Combining soundness and completeness, we conclude that the algorithm returns exactly
\[
\{\, v \mid \dpn{PN}(v)\neq \emptyset \,\},
\]
as required.
\end{proof}

\subsection{Aggregation}

\begin{theorem}
The aggregate-sum algorithm in Algorithm~\ref{alg:aggregate} is correct with respect to its output specification.
\end{theorem}

\begin{proof}
Let
\[
WM=(S,I,F,\Delta)
\]
be the input Weighted NetKAT automaton, and recall that
\[
L(WM)=\iwn{\dw{WN}}.
\]
We must show that Algorithm~\ref{alg:aggregate} returns exactly
\[
\left\{\, \left(v,\sum_{(\tau,w')\in \dw{WN}(v)} w'\right) \;\middle|\; v\in V \right\}.
\]

The proof has two parts.
First, we show that the state-elimination procedure computes, for each input-output packet pair, the sum of the weights of all traces connecting them.
Second, we show that summing the resulting weights over all input and output packets with the same valuation \(v\) yields the desired aggregate for \(v\).

\medskip
\noindent
\textbf{Step 1: correctness of state elimination.}
For any subset \(T\subseteq S\) of states and any pair of packets \(x,y\in Pk\times V\), let
\[
E_T(x,y)
\]
denote the total weight of all paths from \(x\) to \(y\) whose intermediate automaton states lie in \(T\).
We prove, by induction on the number of eliminated states, that the state-elimination algorithm maintains exactly these quantities.

Initially, before any elimination, the transition relation \(\Delta\) already gives the weight of all one-step transitions.
Equivalently, if no intermediate states are allowed, then the corresponding matrix entry records exactly the total weight of all paths from \(x\) to \(y\) with no intermediate states.

Now suppose we eliminate a state \(s\).
Any path from \(x\) to \(y\) whose intermediate states are drawn from \(T\cup\{s\}\) is of one of the following two forms:
\begin{enumerate}
    \item it never visits \(s\), in which case its total contribution is already accounted for by \(E_T(x,y)\); or
    \item it visits \(s\) at least once, in which case it uniquely decomposes into:
    \begin{itemize}
        \item a path from \(x\) to \(s\) with intermediate states in \(T\),
        \item followed by zero or more loops from \(s\) back to \(s\) with intermediate states in \(T\),
        \item followed by a path from \(s\) to \(y\) with intermediate states in \(T\).
    \end{itemize}
\end{enumerate}
Therefore the total weight of all such paths is
\[
E_T(x,y) \;+\; E_T(x,s)\cdot E_T(s,s)^* \cdot E_T(s,y).
\]
This is exactly the usual state-elimination update rule.
Hence, after eliminating \(s\), the new transition weight from \(x\) to \(y\) is precisely the total weight of all paths from \(x\) to \(y\) whose intermediate states avoid \(s\) but may use any previously retained states.

By induction over the elimination order, once all intermediate states have been eliminated, the final transition
\[
\mathit{state\_elimination}(x,y)
\]
is exactly the total weight of all accepting traces from input packet \(x\) to output packet \(y\).
Instantiating \(x=(pk_1,v_1)\) and \(y=(pk_2,v_2)\), we obtain that
\[
\mathit{state\_elimination}((pk_1,v_1),(pk_2,v_2))
\]
is the sum of the weights of all traces in \(L(WM)\) that start from \((pk_1,v_1)\) and end at \((pk_2,v_2)\).

\medskip
\noindent
\textbf{Step 2: aggregation by valuation.}
By Theorem~\ref{thm:correct-compile}, we have
\[
L(WM)=\iwn{\dw{WN}}.
\]
In particular, every accepted weighted trace has the form
\[
(pk_1,v)(pk_2,v)\cdots(pk_n,v),
\]
that is, the valuation component \(v\) is unchanged along the entire trace.
Therefore, every trace contributing to
\[
\mathit{state\_elimination}((pk_1,v_1),(pk_2,v_2))
\]
must satisfy \(v_1=v_2\); if \(v_1\neq v_2\), there is no such trace and the contribution is \(0\).

It follows that, for a fixed valuation \(v\), the quantity
\[
\sum_{pk_1,pk_2\in Pk} \mathit{state\_elimination}((pk_1,v),(pk_2,v))
\]
is exactly the sum of the weights of all traces in \(L(WM)\) whose valuation component is \(v\).
Using again the equality \(L(WM)=\iwn{\dw{WN}}\), this is exactly
\[
\sum_{(\tau,w')\in \dw{WN}(v)} w'.
\]

Hence the algorithm returns precisely
\[
\left\{\,
\left(v,\sum_{pk_1,pk_2\in Pk} \mathit{state\_elimination}((pk_1,v),(pk_2,v))\right)
\;\middle|\;
v\in V
\right\}
=
\left\{\, \left(v,\sum_{(\tau,w')\in \dw{WN}(v)} w'\right) \;\middle|\; v\in V \right\},
\]
which is the required output specification.
\end{proof}

\subsection{Implementation Details}

\begin{algorithm}[t]
\caption{\(\textsc{ClosureAdd}(R)\): Matrix Closure Algorithm on ADDs}
\label{alg:mat-clo}
    \textbf{Input:} An ADD \(R\), where the variables
    \(
    v_0,v_1,\dots,v_n
    \)
    encode the row index of a matrix, and
    \(
    v_0',v_1',\dots,v_n'
    \)
    encode the column index.
    The leaf value of \(R\) is the semiring weight of the corresponding matrix entry.

     \textbf{Output:} An ADD encoding \(R^*\), the Kleene closure of the matrix represented by \(R\).
\begin{enumerate}
    \item If \(R = \textsc{Leaf}(w)\), then return
    \[
    \textsc{Leaf}(w^*).
    \]

    \item Otherwise, \(R\) is an internal node \(\textsc{Node}(v_i,n_1,n_2)\) or \(\textsc{Node}(v_i',n_1,n_2)\). In this case:
    \begin{itemize}
        \item Compute the four block cofactors
        \[
        A \gets \exists v_i,v_i'.\, R[v_i=0,v_i'=0], \qquad
        B \gets \exists v_i,v_i'.\, R[v_i=0,v_i'=1],
        \]
        \[
        C \gets \exists v_i,v_i'.\, R[v_i=1,v_i'=0], \qquad
        D \gets \exists v_i,v_i'.\, R[v_i=1,v_i'=1].
        \]

        \item Recursively compute
        \[
        D^* \gets \textsc{ClosureAdd}(D),
        \]
        \[
        E \gets \textsc{ClosureAdd}(A + B\cdot D^*\cdot C).
        \]

        \item Define
        \[
        F \gets E\cdot B\cdot D^*, \qquad
        G \gets D^*\cdot C\cdot E,
        \]
        \[
        H \gets D^* + D^*\cdot C\cdot E\cdot B\cdot D^*.
        \]

        \item Reconstruct the ADD using Conway's block-matrix star formula:
        \[
        \textsc{Return }\;
        E[v_i=0,v_i'=0]
        + F[v_i=0,v_i'=1]
        + G[v_i=1,v_i'=0]
        + H[v_i=1,v_i'=1].
        \]
    \end{itemize}
\end{enumerate}
\end{algorithm}

\subsection{Batfish Topology}
\begin{figure}
    \centering
\includegraphics[width=1.0\textwidth]{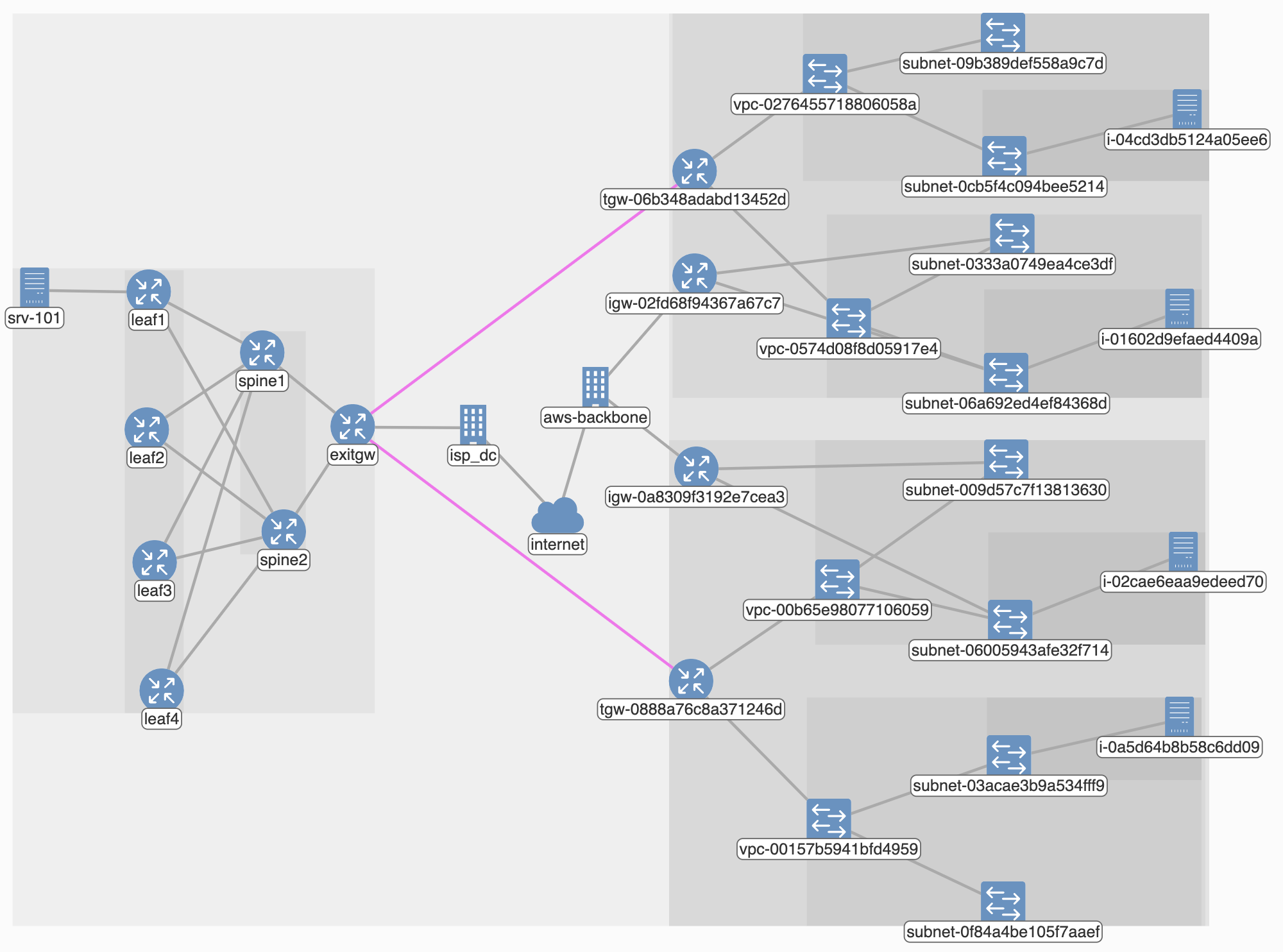}
    \caption{Topology of Hybrid Cloud Network of Batfish}
    \label{fig:batfish}
\end{figure}
\begin{figure}
    \centering
\includegraphics[width=1.0\textwidth]{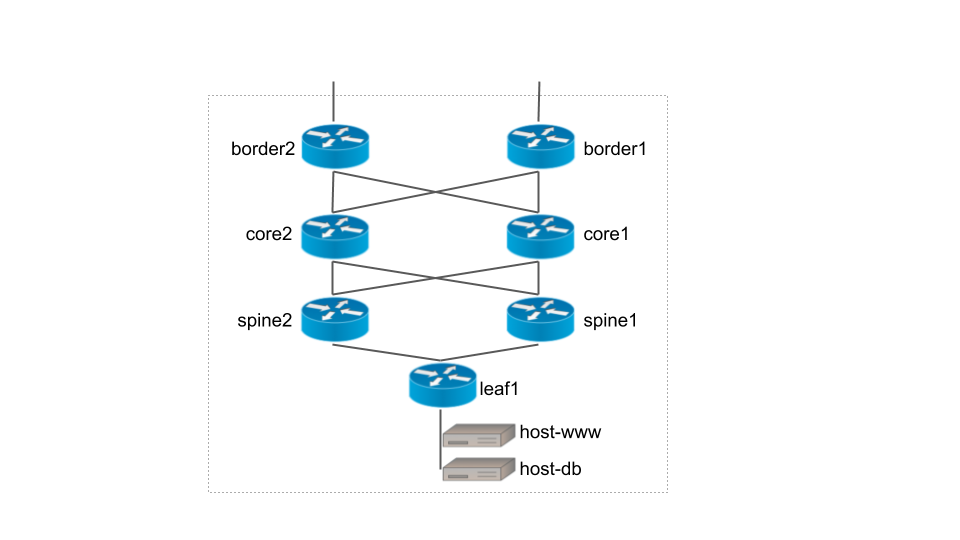}
    \caption{Topology of  Forwarding Change Validation of Batfish}
    \label{fig:batfish2}
\end{figure}

\end{document}